\documentclass[12pt,draftclsnofoot,onecolumn]{IEEEtran}
\usepackage{amsmath}
\usepackage{color}
\usepackage[pdftex]{graphicx}
\usepackage{amssymb}
\usepackage{amsthm}
\usepackage{epstopdf}
\usepackage{mathrsfs}
\usepackage{hyperref}
\usepackage{bbm}
\usepackage{subfig}
\usepackage[dvipsnames]{xcolor}
\usepackage{pgfplots}
\pgfplotsset{compat=1.6}
\usepackage{epstopdf}
\usetikzlibrary{plotmarks}
\newtheorem{thm}{Theorem}
\newtheorem{claim}{Claim}
\newtheorem{defn}{Definition}

\newtheorem{lem}{Lemma}

\graphicspath{{../Figures/}}

\newcommand\given[1][]{\:#1\vert\:}

\newcommand\Bigiven[1][]{\:#1\Big\vert\:}

\title{ Second Order and Moderate Deviation Analysis of a Block Fading Channel with Deterministic and Energy Harvesting Power Constraints }
\author{Deekshith P K, K Gautam Shenoy and Vinod Sharma\\ {Department of ECE, IISc, Bangalore, India.}\\{Email: \{deekshithp, konchady, vinod\}@iisc.ac.in}}

\begin{document}

\maketitle

\begin{abstract}
We consider a block fading additive white Gaussian noise (AWGN) channel with perfect channel state information (CSI) at the transmitter and the receiver. First, for a given codeword length and non-vanishing average probability of error, we obtain lower and upper bounds on the maximum transmission rate. We derive bounds for three kinds of power constraints inherent to a wireless transmitter. These include the canonical peak power constraint and average power constraint, and a time varying peak power constraint imposed by an \emph{energy harvesting device}-a mechanism that powers many modern-day wireless transmitters. The bounds characterize \emph{second order} deviation of finite blocklength coding rates from the channel capacity, which is in turn achieved by \textit{water-filling} power allocation across time. The bounds obtained also indicate the rate enhancement possible due to CSI at the transmitter in the finite blocklength regime. Next, we provide bounds on the optimal exponent with which error probability drops to zero when channel coding rate is simultaneously allowed to approach capacity \emph{at a certain rate}, as the codeword length increases. These bounds identify what is known as the moderate deviation regime of the block fading channel. We compare the bounds numerically to bring out the efficacy of our results.
\end{abstract}
\noindent
\begin{IEEEkeywords}
Finite block length regime, block fading, channel state information, water filling, energy harvesting, moderate deviation.
\end{IEEEkeywords}
\section{Introduction}
\label{Sec:Intro}
 \par In a wireless system, knowledge of instantaneous channel gains at the transmitter can be used to increase \textit{overall data rate} by adaptive modulation and coding, and power control \cite{alouini1999capacity}, \cite{tang2007quality}. For instance, depending on the \emph{fading dynamics}, \emph{capacity} of a flat fading AWGN channel with CSI at the transmitter (CSIT) and  the receiver (CSIR) is enhanced using \emph{water-filling} power control \cite{goldsmith1997capacity}.  However, to capture rate enhancement with CSIT, using metrics like the \emph{water-filling capacity} is fully justified only when there is \emph{no} delay constraint in the system.  Delay requirements are increasingly becoming stringent in 5G applications (see, for instance, \cite{ji2018ultra}), and more refined metrics are required to characterize the performance of wireless systems under such circumstances. In this work, we consider a wireless system with perfect CSIT and CSIR subject to a delay requirement and certain power constraints. For such a system, we characterize the data rate enhancement with CSIT, and \emph{how quickly} the error probability drops, when coding rate under \emph{power adaptation} approaches the channel  capacity, as delay tends to infinity. 
 \par Delay requirements in a wireless network can arise due to quality of service (QoS) constraints arising from the nature of data transmitted, readiness of the users to  pay for extra resources, network architecture, storage limitations etc. All these factors determine at various levels the overall system performance. However, in this work we restrict to delay incurred at the physical layer in sending a codeword to the receiver.
\par With delay constraints imposed at the physical layer, the traditional approach to study rate enhancement due to CSIT (either perfect or imperfect) is to first characterize delay limited capacity, outage capacity or average capacity of the channel (see \cite{el2011network}, Chapter 23 for details on these capacity notions) as per requirement. Once the characterization is in place, the approach is to obtain a power allocation strategy that maximizes the required capacity expression. In this regard, \cite{caire1999optimum} obtains the optimal power allocation that maximises the outage capacity under the assumption of \textit{non-causal} CSIT. In \cite{negi2002delay}, the authors obtain the optimal power allocation scheme maximizing the average capacity with causal CSIT. Nonetheless, the above mentioned notions of capacity may \emph{not} be realistic metrics to evaluate performance of delay sensitive systems with CSIT. This is because, the capacity notions are inherently asymptotic.  
\par In this work, we provide bounds to characterize the rate enhancement due to CSIT over a block fading AWGN channel, at a given codeword length  and average probability of error upto \emph{second order}, using finite block length analyses pioneered in \cite{strassen1962asymptotische}, \cite{hayashi2009information}, \cite{polyanskiy2010channel}. Further, we also characterize rate at which error probability drops to zero when the channel coding rate approaches the capacity at a pre-specified rate, by conducting a \emph{moderate deviation analysis} in the spirit of \cite{altuug2014moderate}, \cite{polyanskiy2010channelmode}. We derive the bounds under three power constraints that are relevant to a wireless transmitter. Our assumption of perfect CSIT is idealistic. However, a wireless transmitter can acquire  instantaneous channel gains either using channel reciprocity  if the wireless system is time duplex  with moderate mobility, or using a dedicated feedback link if the system is frequency duplex with a reasonable round trip delay. In addition, rates obtained under perfect CSIT assumption provide upper bounds for rates achievable without CSIT or with imperfect CSIT and many such characterizations are available in literature \cite{goldsmith1997capacity}, \cite{caire1999optimum}, \cite{negi2002delay}. Also, knowledge of the power control strategies suitable for delay constrained systems sheds insights into how system energy is to be used in such systems. Efficient usage of system energy is essential for energy constrained transmitters which will become prominent { in future wireless networks \cite{mahapatra2016energy}}.  
\par Using finite block length analysis, for instance in the spirit of  \cite{polyanskiy2010channel}, to characterize second order back off from the capacity of a fading channel has attracted a lot of reasearch interest in the past. We will briefly review some of those results and discuss our  contributions in that context. In \cite{polyanskiy2011scalar}, authors consider a scalar coherent fading channel with stationary fading (generalization of block fading)  \emph{without} CSIT and characterize \textit{dispersion} of that channel. In \cite{yang2015optimum}, authors show that the second order optimal power allocation scheme over a \emph{quasi static} fading channel with CSIT and CSIR is \emph{truncated channel inversion}. Quasi static fading corresponds to block fading with a single block. The second order backoff therein is in terms of the number of channel uses within a block. Non-ergodic setting of communication within a single coherence interval is considered also in \cite{molavianjazi2013second} and  authors therein derive a finite block length version of outage capacity.  A MIMO Rayleigh block fading channel with no CSIT and CSIR is considered in \cite{durisi2016short} and achievability and converse bounds are derived for the \textit{short packet communication} regime. A high-SNR normal approximation of the maximal coding rate over a block fading Rayleigh channel without CSIT and CSIR is obtained in \cite{lancho2017high}.  
\par With regard to the moderate deviation analysis of communication channels, results for a discrete memoryless channel is obtained in \cite{altuug2014moderate} and for an AWGN channel in \cite{polyanskiy2010channelmode}. Exact \emph{moderate deviation constants} are characterized for the corresponding channels in these works.  
\par As mentioned, one of the power constraints under which we derive second order rate bounds is imposed by an energy harvesting device. In this context, capacity of a point-to-point \emph{real} AWGN channel and a \textit{fast fading} (i.e., a single channel use within a block) channel with energy harvesting transmitter is obtained in \cite{rajesh2014capacity}.  In the spirit of \cite{hayashi2009information}, \cite{polyanskiy2010channel}, the problem of characterizing the \emph{second order back-off} from the capacity of a real AWGN channel with an energy harvesting transmitter is undertaken in \cite{shenoy2016finite}, \cite{fong2018achievable}. Even though, exact second order back-off coefficient of the channel model is still unknown, bounds on the same can be found in these works. However, such a finite codeword length analysis for a \emph{fading} channel with an energy harvesting transmitter does not seem to be available.
\par Now we summarize the key contributions of this work.
\begin{itemize}
\item We provide lower and upper bounds on the maximal channel coding rate at a codeword length $n$ and average probability of error $\epsilon$, over a  \emph{complex AWGN channel subjected to block fading}, with CSIT and CSIR. This is a refinement and extension of our work on finite state fading block channel in \cite{deekshith2018finite}. Ours is a second order characterization in terms of the number of coherence blocks over which the communication spans. Such a characterization, in terms of the number of coherence blocks, is hitherto not available. Our characterization is in stark contrast to the non-ergodic setting of communication spanning a single block in \cite{yang2015optimum}, wherein the characterization is in terms of number of channel uses within a block. Further, our rate bounds characterize the back-off from the \emph{water-filling} channel capacity in the finite block length regime. We also demonstrate numerically in Section \ref{Sec:Numerical_Egs} that water filling power adaptation at finite blocklength can provide higher rates (upto \emph{second order}) compared notions like to \emph{optimal delay limited power control} \cite{hanly1998multiaccess}. In addition, we derive bounds separately for three kinds of  power constraints  wireless transmitters are normally subjected to namely, peak power constraint, average power constraint and a time varying peak power constraint inherent to an energy harvesting transmitter.
\item  In deriving the bounds, the CSIT assumption makes the analysis involved and non-trivial. In particular, we derive the bounds for block fading channel with general, complex-valued fading gains and hence the \emph{type based} approach in \cite{tomamichel2014second} for finite state discrete memoryless renders infeasible.  Further, in obtaining the upper bounds, the dependence of the channel input on the fading states makes the  corresponding optimization problems difficult to solve. To circumvent this, we derive alternate bounds utilizing the properties of asymptotically optimal power allocation scheme, viz, the water filling scheme, McDiarmid's inequality and \emph{strong approximation of partial sums of i.i.d. random variables} \cite{csorgo2014strong}. We make use of the strong approximation result to bound the probability term in the \emph{relaxed meta converse} (Lemma \ref{Lem:EH_BF_MetaConverse}). This is in contrast to the Berry-Esseen approach for the finite state channel in \cite{tomamichel2014second}. We find that using strong approximation results is conceptually simpler and yields tighter upper bounds for the complex fading channel. An alternate approach is to discretize the state space, make use of the finite state approach and apply the limits to obtain the result for the general fading case. However, such an approach yields loose upper bounds.  The use of strong approximation results to obtain finite block length bounds is novel in such a setting and can be of independent interest.
\item We derive moderate deviation bounds in the spirit of \cite{altuug2014moderate}, \cite{polyanskiy2010channelmode}, for the block fading channel under consideration. These bounds characterize the optimal rate at which the probability of error converges to zero when the channel coding rate is allowed to scale at a particular rate as the codeword length increases. Such an analysis has not been previously undertaken for the canonical or the energy harvesting version of  a block fading channel.   
\end{itemize}  
\par This paper is organized as follows. In Section \ref{Sec:Model_Notation}, we introduce the system model and notation. We provide  lower and upper  bounds on the maximal channel coding rate with peak codeword power constraint in Section \ref{Sec:PeakPower_Trans}. Next, in Section \ref{Sec:AvgPower_Trans}, we  provide  lower and upper  bounds on the maximal channel coding rate with average codeword power constraint. Similar bounds for the case when the transmitter is harvesting energy from the environment are provided in Section \ref{Sec:EH_Trans}. Moderate deviation rate is obtained in Section \ref{sec_mod_dev}. In Section \ref{Sec:Numerical_Egs}, we compare the bounds numerically and demonstrate the utility of the bounds derived. We conclude in Section \ref{Sec:Conclusion}. Proofs are delegated to the appendices.  
\section{Model and Notation}
\label{Sec:Model_Notation}
\subsection{Channel Model}
\label{SubSec:Channel_Model}
We consider a point-to-point, discrete time, frequency non-selective block fading channel subject to additive, circularly symmetric complex Gaussian noise at the receiver. The probability density function (pdf) of the additive noise random variable with covariance matrix $\left[\begin{smallmatrix}\sigma_N^2/2 & 0\\ 0 & \sigma_N^2/2\end{smallmatrix}\right]$ is denoted as $\mathcal{CN}\left(0,\sigma_N^2\right)$. The additive noise process is  independent and identically distributed (i.i.d.) across channel uses. A \emph{slot} refers to the time period between successive channel uses; the period between $i$\textsuperscript{th} and $(i+1)$\textsuperscript{th} channel use is slot $i$. A \emph{block} refers to a time period of duration $T_c$  (channel coherence time) over which the gain of the underlying wireless channel remains constant. Let the delay constraint  imposed, by some application, e.g., voice, on the communication at the physical layer be $D$ time units. In this work, we restrict exclusively to the time for information transmission. Thus, let $D$ be the time left for transmission after taking care of the transmission of training symbols for channel state estimation, synchronization, transmission of control information etc. We assume the transmitter has perfect knowledge of the  channel state. Considering for convenience $D$ as an integer multiple of $T_c$, $B={D}/{T_c}$ is the number of blocks over which the communication spans. Let $n_c$ denote the number of times the in-phase and quadrature channels are used within a block. Then, the number of complex channel uses for the whole of communication, or equivalently, the codeword length $n=Bn_c$. 
\par The channel gain or the fading coefficient in block $b$ is denoted as $H_b\in \mathbb{C}$, the set of complex numbers, such that $\mathbb{E}\left[|H_b|^2\right]=\sigma_H^2<\infty$. Here, $\mathbb{E}[\cdot]$ denotes the expectation operator and $|\cdot|$ denotes absolute value. We assume the channel gains are i.i.d. across blocks and is independent of the additive noise process. Let $F_{\text{H}}$ (known to both the transmitter and the receiver) denote the cumulative distribution function (cdf)  common to all $H_b$, $b\in[1:B]$. (Here, for natural number $m$, $[m:m+\ell-1]$ denotes the set of $\ell$ consecutive natural numbers from $m$.) The instantaneous channel gains are assumed to be known to the transmitter  as well as the receiver and the transmitter gets to know them only causally. We refer to this as the full CSIT and CSIR assumption. Even though, in practice, a wireless transmitter will \emph{not} have perfect channel state information, rates obtained under such an assumption provide upper bounds for rates with partial CSIT. Further, our analysis can be considered as the first step in  understanding the possible reduction in rate under various partial CSIT assumptions like quantized state, minimum mean square error estimate of the state etc. Also, for the energy harvesting case considered later, our analysis yields results for no CSIT and full CSIT cases either of which are otherwise unknown.
\par Let $X_{(b-1)n_c+k}$ denote the channel input corresponding to {the} $k^{\text{th}}$ channel use in {the} $b^{\text{th}}$ block, where, $k \in [1:n_c]$, $b \in [1:B]$. For convenience, from here on, $[b,k]\triangleq (b-1)n_c+k$.  Let $Z_{[b,k]}$ and $Y_{[b,k]}$ denote the corresponding noise variable and the channel output, respectively. Then, $Y_{[b,k]}=H_{b}X_{[b,k]}+Z_{[b,k]}.$ For this channel model, if the number of blocks $B$ tends to infinity (and hence $D\rightarrow \infty$), it is well known that the channel capacity is given by 
\begin{equation}
\label{Eqn:Cap_BlockFadeCha}
\mathbf{C}\left(\bar{P}\right)\triangleq\mathbb{E}\left[\log\left(1+\frac{|H_1|^2\mathcal{P}_{\text{WF}}\left(|H_1|\right)}{\sigma_N^2}\right)\right],
\end{equation}
 where, $\mathcal{P}_{\text{WF}}(|H_1|)\triangleq \Big(\lambda-\frac{\sigma_N^2}{|H_1|^2}\Big)^+,$ $(\cdot)^+=\max(0,\cdot)$ (see \cite{el2011network}, Section 23.2 wherein the result is given for a \emph{real} block fading channel and \cite{tse2005fundamentals}, Section 5.4.6 for a \emph{complex} block fading channel with $n_c=1$). Further, $\mathbf{C}\left(\bar{P}\right)$ is the capacity under an \emph{average power constraint} and $\lambda$ is obtained by solving the equation $\mathbb{E}\left[\mathcal{P}_{\text{WF}}\left(|H_1|\right)\right]=\bar{P}.$ Here, $\mathcal{P}_{\text{WF}}(\cdot)$ is called the water-filling power allocation with average power $\bar{P}$.
\subsection{Transmitter Model}
\label{SubSec:Trans_Model}
\par Let $S$ be the message to be transmitted, chosen randomly uniformly from $[1:M]$. The message is encoded into a codeword $\mathbf{X}'=(X_1',\hdots,X_n')$, possibly randomly chosen. Let $\mathbf{X}=(X_1,\hdots,X_n)$ be the corresponding channel input vector. As we explain below, $\mathbf{X}'$ and $\mathbf{X}$ can possibly be different owing to power control, energy unavailability etc. Also, in a wireless transmitter, there can be various kinds of constraints on $\mathbf{X}$. Next, we explain these constraints.  
\par In practice, a wireless transmitter is subjected to various kinds of power constraints. Inherently, there is a restriction on the maximum power that can be expended, due to circuitry limitations, regulatory requirements etc. We refer to such a constraint as the \emph{peak power constraint}. If the transmitter has knowledge about the gain of the wireless channel, it can adjust the transmit power according to the channel gain. However, the average power that can be expended over all channel gains is usually bounded. This constraint, referred to as the \emph{average power constraint},  corresponds to the long term power utilization efficiency of the system. Though in reality these constraints are simultaneously present, in this work, as is the usual practice (see \cite{polyanskiy2011scalar} for an instance of a peak power constrained transmitter and \cite{yang2015optimum} for an average power constrained transmitter), we study them only in isolation. In addition, we also consider an \emph{energy harvesting transmitter} motivated by the fact that such communication systems are increasingly becoming popular \cite{ku2016advances}. We will consider this in more detail later.
\subsubsection{Peak Codeword Power Constrained Transmitter}
\label{SubSubSec:PeakCodePow_Trans}
Given the message $S$ to be transmitted and the fading gains $ (H_1,\hdots,H_b)\equiv H^{(b)}$ till block $b\in[1:B]$ (we denote the corresponding realizations as $h^{(b)}$), at ${[b,k]}$\textsuperscript{th} channel use, the encoder chooses the codeword symbol $X_{[b,k]}'\left(S,H^{(b)}\right) \equiv X_{[b,k]}'$. In certain circumstances, especially while proving lower bounds on the maximal achievable rate, it is convenient to consider that the transmitter encodes the message independent of channel gain realizations and later adapt the power of codeword symbols using a \emph{power controller}. Such a scheme is attractive from an implementation perspective as well, as it eliminates the need for variable rate coding by multiplexing several codebooks (see \cite{caire1999capacity}, Section IV). In such a case, $X_{[b,k]}'$ will be different from $X_{[b,k]}\left(S,H^{(b)}\right)\equiv X_{[b,k]}$, the corresponding channel input symbol. When the encoding is dependent on the causally available channel gains and hence, there is no separate power control, $\mathbf{X}=\mathbf{X}'$ . We say that a transmitter has \emph{peak codeword power constraint} (abbreviated as \textbf{PP}), if, for each $m\in[1:M]$ and  every realization $ \mathbf{h}\triangleq (h_1,\hdots,h_B)$ of the fading vector $ \mathbf{H}\triangleq(H_1,\hdots,H_B)$, the channel input vector $\mathbf{X}$ satisfies 
\begin{equation}
\label{Eqn:Defn_PCP_Constr}
 \sum_{b=1}^{B}\sum_{k=1}^{n_c}\left|X_{[b,k]}\left(m,h^{(b)}\right)\right|^2 \leq  Bn_c\bar{P}.
\end{equation} 
The above inequality holds almost surely (a.s.) with respect to the underlying joint distribution.
 \subsubsection{Average Codeword Power Constrained Transmitter}
 We say that a transmitter has \emph{average codeword power constraint} (\textbf{AP}), if for each $m\in[1:M]$, the channel input vector $\mathbf{X}$ satisfies
\begin{equation}
\label{Eqn:Defn_ACP_Constr}
\mathbb{E}\left[\sum_{b=1}^{B}\sum_{k=1}^{n_c}\left|X_{[b,k]}\left(m,H^{(b)}\right)\right|^2\right] \leq  Bn_c\bar{P}.
\end{equation} 
Here the expectation is with respect to the fading gain vector $\mathbf{H}$.
\label{SubSubSec:AvgCodePow_Trans}
\subsubsection{Energy Harvesting Transmitter}
\label{SubSubSec:EH_Trans}
\par In this case, the transmitter is connected to an energy harvesting device that provides power to transmit codeword symbols. Prior to $[b,k]$\textsuperscript{th} channel use, the device  harvests energy $E_{[b,k]}\geq 0$ from its ambient environment. The harvested energy $E_{[b,k]}$  is made available for transmission in the same slot itself. We assume the amount of energy harvested is i.i.d. across slots, with  mean $\overline{E}$, variance $\sigma_E^2$ and $\mathbb{E}[E_{[1,1]}^4]<\infty$.
 \par The harvesting device is equipped with an energy buffer (equivalently, an \emph{accumulator} or a battery) which we simply refer to  as a buffer. The buffer stores the residual energy after  $[b,k]$\textsuperscript{th} transmission for future use.  This model is often referred to as the \emph{harvest-use-store} model \cite{rajesh2014capacity}. We assume the buffer has  infinite energy storage capacity. This is a simplifying assumption and is usually made in literature \cite{rajesh2014capacity}, \cite{fong2016non}. Prior to $[b,k]\textsuperscript{th}$ transmission, given $S,~H^{(b)}$ and energy realizations $\left(E_{[1,1]},\hdots,E_{[b,k]}\right)\equiv E^{([b,k])}$, the encoder chooses  $X_{[b,k]}'\left(S,E^{([b,k])},H^{(b)}\right)\equiv X_{[b,k]}'$ as the $[b,k]$\textsuperscript{th} codeword symbol. The exact symbol transmitted in $[b,k]$\textsuperscript{th} slot could be different from $X_{[b,k]}'$. This is due to the fact that the transmitter is harvesting  energy and hence, may not have enough energy to transmit $X_{[b,k]}'$ in that slot. The corresponding channel input is $ X_{[b,k]}$, $X_{[b,k]}\in\mathbb{C}$. 
\par Let $A_{[b,k]}$ denote the energy available in the buffer  at the beginning of  $[b,k]$\textsuperscript{th} transmission.  We assume that the  buffer is initially empty, i.e. $A_{[1,1]} = 0$. Since we consider the harvest-use-store model, the total energy available to the transmitter at the beginning of $[b,k]$\textsuperscript{th} transmission is $E_{[b,k]}+A_{[b,k]}\triangleq \hat{A}_{[b,k]}$.
The buffer state $A_{[b,k]}$ evolves according to $A_{[b,k+1]}=\hat{A}_{[b,k]} - |X_{[b,k]}|^2,~ k\in[1:n_c-1],~n_c\neq 1,$  and $A_{[b+1,1]}=\hat{A}_{[b,n_c]} - |X_{[b,n_c]}|^2,~  k=n_c,$
where $|X_{[b,k]}|^2\leq \hat{A}_{[b,k]}$. This recursive relation and  $A_{[b,k]}\geq 0$ a.s. for all $b$, $k$ yields the energy harvesting constraint (\textbf{EH})
\begin{equation}
\label{Eqn:Defn_EH_Constr}
\sum_{\ell=1}^{b}\sum_{i=1}^{k}\left|X_{[\ell,i]}\right|^2\leq \sum_{\ell=1}^{b}\sum_{i=1}^{k}E_{[\ell,i]},~\text{a.s.}
\end{equation}
\par The decoder used in obtaining  our lower bounds is  same as that used in the $\beta\beta$ achievability bound \cite{yang2018beta}. We require the average probability of error $\mathbb{P}\left[\psi(\mathbf{Y},\mathbf{H}) \neq S\right]\leq \epsilon$,  $0<\epsilon<\frac{1}{2}$.
\par \emph{Goal}: Let the maximum size of the codebook with codeword length $n$ and average probability of error  $\epsilon$, with $\textbf{PP}$ constraint be denoted as $M_p^*(n,\epsilon,\bar{P})\equiv M_p^*$. The corresponding maximal coding rate (in bits per channel use) $R_p^*(n,\epsilon,\bar{P})\equiv R_p^*=n^{-1}\log M^*_p$. Similarly, for transmitters with \textbf{AP} and \textbf{EH} constraint we define $M_a^*$ and $R_a^*$, and $M_e^*$ and $R_e^*$, respectively. Our primary goal is to obtain tractable lower and upper bounds for $R_p^*$, $R_a^*$ and $R_e^*$. In addition, we also conduct a moderate deviation analysis to characterize the \emph{exponential rate} at which error probability drops to zero, when  $R_p^*,~R_a^*$ and $R_e^*$ are allowed to converge to capacity at a pre-specified 
rate.
\par We conclude this section by making note of the notation that we use in the rest of the paper. We fix all logarithms to the base $2$. However, $\exp(\cdot)$ denotes the exponent $e$. Let $ C(x)\triangleq \log\big(1+x\big)$, $\mathcal{L}(x)\triangleq  {x}/{1+x}$, {$V(x)\triangleq  {x(2+x)}/{(1+x)^2}$ }. The set of integers, positive integers, real and positive real numbers are denoted as $\mathbb{Z}$, $\mathbb{Z}_+$ $\mathbb{R}$ and $\mathbb{R}_+$, respectively. The indicator function of an \textit{event} $A$ is $\mathbbm{1}_A$ and  $A^\ell \equiv A_1\times\hdots \times A_\ell$.  For $p=1,~2$, given $\mathbf{a},~\mathbf{c}$ belonging to $\mathbb{C}^\ell$, $||\mathbf{a}||_p$ denotes the $p-$norm of $\mathbf{a}$ and $\langle{\mathbf{a},\mathbf{c}}\rangle$ the inner product of $\mathbf{a}$ and $\mathbf{c}$. In particular, $||a||_2$ is simply denoted as $||a||$. For $\alpha\in\mathbb{R}$, $\mathbf{a}^\alpha \triangleq (a_1^\alpha,\hdots,a_\ell^\alpha)$ and $\alpha\cdot \mathbf{a}=\left(\alpha a_1,\hdots,\alpha a_n\right)$. For a function $f$ defined on $\mathbb{C}$, $(f(a_1^\alpha),\hdots,f(a_\ell^\alpha))$ is sometimes compactly expressed as $f(\mathbf{a}^\alpha)$. If $\ell=Bn_c$, $\mathbf{a}_b=(a_{[b,1]},\hdots,a_{[b,n_c]})$ and $a_{[b,k]}=a_{[b,k],1}+ja_{[b,k],2}$. For $q\in\mathbb{R}_+,~\ell\in\mathbb{Z}_+$, let $\mathcal{X}_{\ell}(q)=\big\{\mathbf{x}: ||\mathbf{x}||_2=\sqrt{\ell q}\big\}$. At times, we denote the sum of  $a_1,\hdots,a_\ell$, $ \sum_{k=1}^{\ell}a_k$ as $\mathbf{S}_\ell(a_1)$. The notation $\mathbf{U}\perp\mathbf{V}$ denotes that the random vectors are independent. The variance of a random variable $U$ is denoted as $\mathbb{V}[U]$ and the expectation with respect to $\mathbf{U}$ as $\mathbb{E}_{\mathbf{U}}[\cdot]$. Also, $\mathbf{U}\sim F$ denotes $\mathbf{U}$ is distributed according to the distribution $F$. The function $\Phi(\cdot)$ denotes the cdf of a standard Gaussian random variable, $\phi(\cdot)$ denotes the corresponding pdf, $\Phi^{-1}(\cdot)$ denotes the inverse cdf and $X\sim \mathcal{N}(a,b)$ denotes that the random variable $X$ has normal distribution with mean $a$ and variance $b$. We use the standard Bachman-Landau asymptotic notation $o(\cdot),~O(\cdot)$. The notation $\overline{\lim}\equiv \limsup,$ $\underline{\lim}\equiv \liminf$. The notation $\stackrel{D}{=}$ means equivalence in distribution. 

\section{Peak Power Constrained Transmitter}

\label{Sec:PeakPower_Trans}

\par  In this section, we provide lower and upper bounds on $R^*_p$, the maximal coding rate under peak codeword power constraint mentioned in \eqref{Eqn:Defn_PCP_Constr}. With $\mathcal{P}_{\text{WF}}(\cdot)$ and $\lambda$ defined as in \eqref{Eqn:Cap_BlockFadeCha}, let $G_b \triangleq |H_b|\sqrt{\mathcal{P}_{\text{WF}}(|H_b|)} \big /{\sigma_N}$. For $\alpha\in(0,1)$,  $c_\epsilon\triangleq \sqrt{{2(n_c+1)\log (1/\alpha\epsilon)}/{\log e}}$ and
\begin{equation}
\label{Eqn:Defn_V_BF_PP}
V_{\text{BF}}(\bar{P}) \triangleq \mathbb{E}\left[V\left(G_1^2 \right)\right]+n_c\mathbb{V} \left[C\left(G_1^2\right)\right]+\mathbb{V}\left[\mathcal{L}\left(G_1^2\right)\right].
\end{equation}
With the above notation in place, we have the following result that provides  lower and upper bounds on $R_p^*$.
 \begin{thm}
\label{Th:PeakPow_Bnds}
Consider the block fading channel described in Section \ref{SubSec:Channel_Model}  with a peak power constrained transmitter as in \ref{SubSubSec:PeakCodePow_Trans}. For $n$ large enough and any $\alpha\in(0,1)$, the maximal rate $R^*_p$ is lower bounded as
\begin{equation}
\label{Eqn:Bnd_LB_R_p_star_Th_PP}
R_p^*\geq \mathbf{C}(\bar{P})-\frac{c_\epsilon}{\sqrt{n}}+\sqrt{\frac{V_{\text{BF}}(\bar{P})}{n}}\Phi^{-1}\left(\frac{(1-\alpha)\epsilon}{2}\right)+O\left(\frac{\log n}{n}\right).
\end{equation}
Further, $R^*_p$ is upper bounded as
\begin{equation}
\label{Eqn:Bnd_UB_R_p_star_Th_PP}
R^*_p\leq \mathbf{C}(\bar{P})+\sqrt{\frac{V_{\text{BF}}(\bar{P})}{n}}\Phi^{-1}(\epsilon)+o\left(\frac{1}{\sqrt{n}}\right),
\end{equation}
  \end{thm}
\begin{proof}
See Appendix \ref{App:Proof_PeakPow}.
\end{proof}
\subsection{Discussion of results}
\label{SubSec:PeakPower_Trans}
First, note that, with \emph{no} CSIT and \textbf{PP} constraint, the \emph{optimal}  coefficient of $1/\sqrt{n}$ term, i.e., the channel dispersion or optimal second order coefficient, is obtained by replacing $G_1$ in the definition of $V_{\text{BF}}(\bar{P})$ in \eqref{Eqn:Defn_V_BF_PP} with $|H_1|$  \cite{polyanskiy2011scalar}. If we restrict to the second order approximation of the bounds on $R_p^*$ (i.e., discarding the terms whose order is \emph{higher} than $1/\sqrt{n}$ on the RHS of  \eqref{Eqn:Bnd_LB_R_p_star_Th_PP} and \eqref{Eqn:Bnd_UB_R_p_star_Th_PP}), in Section  \ref{SubSec:Num_Egs_Non_EH}, we will numerically evaluate the upper bound and show that the bound could be significantly lower than that predicted by the capacity expression in \eqref{Eqn:Cap_BlockFadeCha}, depending on the value of system parameters. Similarly, we will also show that the rate predicted by the second order approximation of the lower bound in \eqref{Eqn:Bnd_LB_R_p_star_Th_PP} obtained by means of power control, can be higher than that without power control in \cite{polyanskiy2011scalar}. This will also point to the rate enhancement possible, depending on the value of various system parameters, with CSIT in place. However, we also note that second order coefficients in our lower and upper bounds do not match. Hence we also study the proximity of the bounds numerically in \ref{SubSec:Num_Egs_Non_EH}. 
\par Next, we consider the \textbf{AP} constraint case.
\section{Average Power Constrained Transmitter}

\label{Sec:AvgPower_Trans}
The benefit of power control is more pronounced when the system has an average power power constraint rather than a peak power constraint. In this section, we provide lower and upper bounds on the achievable rate $R_a^*$ subject to \textbf{AP} constraint. Towards that, we have the following definition. A power allocation function is a mapping  $\mathcal{P}:\mathbb{C}\mapsto \mathbb{R}_+$ that defines the amount of power expended on channel states. With other notation as in the previous section, we have the following result.
\begin{thm}
\label{Th:AvgPow_Bnds}
Consider the block fading channel described in Section \ref{SubSec:Channel_Model}  with an average power constrained transmitter as in \ref{SubSubSec:AvgCodePow_Trans}. For $n$ large enough, the maximal rate $R^*_a$ is lower bounded as 
\begin{equation}
\label{Eqn:Bnd_LB_R_a_star_Th_AP}
R_a^*\geq  \mathbf{C}(\bar{P})+\sqrt{\frac{V_{\text{BF}}(\bar{P})}{n}}\Phi^{-1}\left(\epsilon\right)+O\left(\frac{\log n}{n}\right).
\end{equation}
Further, if we restrict to the class of power allocation policies with $\mathbb{E}\left[\mathcal{P}^{2+\delta}(H_1)\right]<\infty$, for some $\delta>0$, $R^*_a$ is upper bounded as
\begin{equation}
\label{Eqn:Bnd_UB_R_a_star_Th_AP}
R_a^*\leq  \mathbf{C}(\bar{P})+\sqrt{\frac{V_{\text{BF}}(\bar{P})}{n}}\Phi^{-1}\left(\epsilon\right)+o\left(\frac{1}{\sqrt{n}}\right).
\end{equation}
 \end{thm}
\begin{proof}
See Appendix \ref{App:Proof_AvgPow}.
\end{proof}
\subsection{Discussion of results}
\label{SubSec:AvgPower_Trans}
The above bounds illustrate that water filling power allocation is second order optimal (under a mild additional assumption on the power allocation function) for the case of \textbf{AP} constraint under the finite block length and non-vanishing error probability setting. From the proof in Appendix \ref{App:Proof_AvgPow}, we note that for a block fading channel model with finite number of fading states, the upper bound in \eqref{Eqn:Bnd_UB_R_a_star_Th_AP} holds without the moment restriction. We also note that the capacity achieving water filling power allocation with average power $\bar{P}$ is peak power constrained (see \eqref{Eqn:Cap_BlockFadeCha})  and hence the additional moment constraint is naturally met. In addition, our bounds also establishes that, as in the capacity case, non-causal CSIT is not helpful in improving the rate upto second order (again, under a mild assumption on power allocation function). 
\par We illustrate the rate enhancement possible due to CSIT under \textbf{AP} constraint in Section \ref{SubSec:Num_Egs_Non_EH}.
\section{Energy Harvesting Transmitter}
\label{Sec:EH_Trans}
In this section, our goal is to obtain  tractable lower and upper bounds on the achievable rate $R_e^*$, with codeword length $n$, average probability of error  $\epsilon$, satisfying energy harvesting constraints with average harvested energy $\overline{E}$.
We begin with explaining an encoding-decoding scheme that we use in the proof of our lower bound.  
\subsection{Save and Transmit Scheme}
\label{SubSec:Save_Transmit}
\par Corresponding to each $m\in[1:M]$, the encoder generates an $n$ length codeword $\mathbf{X}'(m)\triangleq \left(X_1'(m),\hdots,X_n'(m)\right)$ of i.i.d. random variables $X_i'(m)$ distributed according to $\mathcal{CN}(0,1)$, where $n=Bn_c$ and $i\in[1:n]$. This encoding scheme is used in conjunction with a \emph{power controller} and a \emph{transmission scheme} called the \emph{Save and Transmit} scheme. The transmission scheme is known to achieve the capacity of a real AWGN channel  \cite{ozel2012achieving}. 
\par Save and Transmit scheme consists of two phases. To transmit a codeword of length $n$, the transmitter does not transmit in the first $N_n$ harvesting slots, where $N_n$ is appropriately chosen. This phase is called the energy saving phase. After the first $N_n$ slots, the buffer contains some energy and the codeword $\mathbf{X}'(S)$ is transmitted in the next $n$ slots, where $S$ is the message selected. This is the \emph{transmission phase}. If  $[b,k]$\textsuperscript{th} symbol to be transmitted requires more energy than what is available i.e., $|{X'}_{[b,k]}|^2> \hat{A}_{[b,k]}$,  we refer to the slot as being in an \emph{energy outage}. The channel output in the transmission phase is $Y_{[b,k]}=H_bX_{[b,k]}'\mathbbm{1}\{\hat{A}_{[b,k]}\geq |{X'}_{[b,k]}|^2\}+Z_{[b,k]}$.
\par   The decoder $\psi:\mathbb{C}^{(n)}\times\mathbb{C}^{(B)} \mapsto [1:M]$ upon receiving $\mathbf{Y}\triangleq \left(Y_{[1,1]},\hdots,Y_{[B,n_c]}\right)$, obtains $\psi(\mathbf{Y},\mathbf{H})\equiv \hat{S}$, an estimate of the message $S$ transmitted. The decoder used in obtaining  our lower bound is  same as that used in obtaining the lower bounds in Theorem \ref{Th:PeakPow_Bnds} and Theorem \ref{Th:AvgPow_Bnds}.
 \par We recall certain prior results that we use in deriving the bounds. The capacity of a fast fading channel (i.e., $n_c=1$) with an energy harvesting transmitter was previously characterized in \cite{rajesh2014capacity} (see Theorem 3 therein). Using the same analysis as in \cite{rajesh2014capacity}, with $\mathbb{E}[\mathcal{P}_{\text{WF}}(|H_1|)]=\overline{E}$, it is easy to see that the capacity of the block fading channel under consideration is given by 
\begin{equation}
\label{Eqn:Cap_EH_BFC}
\mathbf{C}\left(\overline{E}\right)=\mathbb{E}\left[\log \left(1+\frac{|H_1|^2\mathcal{P}_{\text{WF}}(|H_1|)}{\sigma_N^2}\right)\right].
\end{equation}  

\par The following lemma from \cite{shenoy2016finite} will aid in tackling the proof of the  lower bound by decoupling the channel statistics from the analysis of the energy outage event described in Section \ref{SubSec:Save_Transmit}.

\begin{lem}
\label{Lem:EH_Error_Control}
Consider a communication channel with an energy harvesting transmitter. Suppose the input to the channel is generated i.i.d. with mean zero and variance $\overline{E}$, and the energy arrivals are i.i.d. with mean $\overline{E}$ and variance $\sigma_E^2$. Let $\epsilon$ be the allowed probability of error (maximal or average) for the system. Then, given $0<\alpha<1$, there exists a Save and Transmit scheme, with the duration of transmission phase  $n$ slots and the duration of the saving phase  $N_n=K_{\epsilon,\alpha}\sqrt{n}$ slots. Further, with $K_{\epsilon,\alpha}=\sqrt{\frac{4(2\overline{E}^2+\sigma_E^2)}{(1-\alpha)\epsilon\overline{E}^2}}$, the probability of error in decoding $\mathbb{P}\left[\hat{S}\neq S\right]\equiv \epsilon_{n}$  referred to as the {non-energy harvesting error} is bounded as
\begin{equation}
\label{Eqn:Bnd:Non_EH_Error}
	\epsilon_{n} \leq \alpha \epsilon - \frac{4\sigma_E^2}{K_{\epsilon,\alpha}\overline{E}^2\sqrt{n}}.
\end{equation}

\end{lem}
\par Next lemma is a version of the meta-converse (derived from \cite{polyanskiy2010channel}, \cite{fong2018achievable}) that is suitably adjusted to incorporate the effects of the stochastic nature of harvested energy and fading gains. We use this result in our proof of the upper bound.
\begin{lem}
\label{Lem:EH_BF_MetaConverse} 
Consider a block fading channel as in Section \ref{SubSec:Channel_Model} with an energy harvesting transmitter as in Section \ref{SubSubSec:EH_Trans}. Let the codewords be generated according to $\mathbb{P}_{\mathbf{X}'|\mathbf{E,H}} \in \mathcal{P}(\mathbf{E,H})$, where $\mathcal{P}(\mathbf{E,H})$ is the set of input distributions that satisfy the energy harvesting constraints described in Section \ref{SubSubSec:EH_Trans}. Then, any code with $M$ codewords of codeword length $n$ and average probability of error $\epsilon$, satisfies
	\begin{equation}
	\label{Eqn:Meta_Converse}
		M \leq \sup_{\mathbb{P}_{\mathbf{X'|E,H}}\in \mathcal{P}(\mathbf{E,H})} \frac{\gamma_n}{\mathbb{P}\left[\log\frac{d\mathbb{P}_{\mathbf{Y|X',H}}}{d\mathbb{Q}_\mathbf{Y|H}} \leq \log \gamma_n\right]-\epsilon} 
	\end{equation}
	for any $\gamma_n>0$ and any auxiliary channel $\mathbb{Q}_{\mathbf{Y|H}}$, whenever the RHS of \eqref{Eqn:Meta_Converse} is non-negative. Here, the probability in the denominator is with respect to the distribution $\mathbb{P}_{\mathbf{E}}\mathbb{P}_{\mathbf{H}}\mathbb{P}_{\mathbf{X'|E,H}}\mathbb{P}_{\mathbf{Y|X',H}}$.
\end{lem}
\par Now, we use the above lemmas to obtain the following result. 
\begin{thm}
\label{Thm:R_BF_LBUB_FBL_EH}
Consider the channel in Section \ref{SubSec:Channel_Model} and energy harvesting model in Section \ref{SubSubSec:EH_Trans}. With $\mathbf{C}\left(\overline{E}\right)$ as in \eqref{Eqn:Cap_EH_BFC}, $K_{\epsilon,\alpha}$ as in \eqref{Eqn:Bnd:Non_EH_Error}, $
V_{EF}'\left(\overline{E}\right)\triangleq \mathbb{E}\big[\mathcal{L}\big(G_1^2\big)\big]+n_c\mathbb{V} \big[C\big(G_1^2\big)\big]+\mathbb{V}\big[\mathcal{L}\big(G_1^2\big)\big]$ and
\begin{equation}
\label{Eqn:Defn_2nd_OrdCoeff2_LB_EH}
V_{\epsilon,\alpha}\left(\overline{E}\right) \triangleq \sqrt{V_{EF}'\left(\overline{E}\right)}\Phi^{-1}\left(\alpha\epsilon\right)-K_{\epsilon, \alpha}\mathbf{C}\left(\overline{E}\right).
\end{equation}
Then, the maximal coding rate $R_e^*$ satisfies  the following bounds.
\begin{enumerate}
\item For any $0 < \alpha < 1$ and sufficiently large $n$, 
\begin{equation}
\label{Eqn:R_BF_LB_FBL_EH}
R_e^*\geq \mathbf{C}\left(\overline{E}\right)+{V_{\epsilon,\alpha}\left(\overline{E}\right)}/{\sqrt{n}} + o(1/\sqrt{n}).
\end{equation}
\item  For sufficiently large $n$, with 
$V_1(x)\triangleq \mathcal{L}^2(x)+{\sigma_E^2}/{\lambda^2}$, for $x\in\mathbb{R}_+$, $\lambda$ as in \eqref{Eqn:Cap_EH_BFC} and $
V_{\text{EF}}''(\overline{E})\triangleq \mathbb{E}\big[V_1\big(G_1^2\big)\big]+n_c\mathbb{V} \big[C\big(G_1^2\big)\big]+\mathbb{V}\big[\mathcal{L}\big(G_1^2\big)\big],$
\begin{equation}
\label{Eqn:R_BF_UB_FBL_EH}
R_e^*\leq \mathbf{C}\left(\overline{E}\right)+\sqrt{V_{\text{EF}}''(\overline{E})/n}\Phi^{-1}(\epsilon) + o(1/\sqrt{n}).
\end{equation}
\end{enumerate}
\end{thm}

\begin{proof}
We provide a sketch of arguments involved in the proof of the lower bound in \eqref{Eqn:R_BF_LB_FBL_EH}. As in the case of \textbf{PP} and \textbf{AP} constraints, we decouple coding and power control. The encoder generates $M$ i.i.d. codewords with the symbols distributed independently $\mathcal{CN}(0,1)$. The codebook is shared with the receiver. The transmitter adopts  the Save and Transmit scheme described in Section \ref{SubSec:Save_Transmit}. Thus, the transmitter waits for first $N_n$ slots where, $N_n$ is fixed as in Lemma \ref{Lem:EH_Error_Control}. After $N_n$ slots, the  transmitter starts sending the codeword corresponding to the message chosen. Codeword symbol  $X_{[b,k]}'$ in slot $[b,k]$ is a complex number with two components $X_{[b,k],1}'$ and $X_{[b,k],2}' $, to be sent across the in-phase and quadrature channels, respectively. Each component is passed through a power controller and the output of the power controller is $X_{[b,k],i}=\sqrt{\mathcal{P}_{\text{WF}}(|H_b|)}X_{[b,k],i}'$, for $i=1,2$. If $X_{[b,k],1}^2+X_{[b,k],2}^2\leq\hat{ A}_{[b,k]}$, both symbols are transmitted, else no transmission happens in that slot. We choose $N_n$ as in Lemma \ref{Lem:EH_Error_Control} so that, the non-energy harvesting error is upper bounded as in \eqref{Eqn:Bnd:Non_EH_Error}. We use this estimate to perform a similar analysis as in  Appendix \ref{App:Proof_PeakPow} and Appendix \ref{App:Proof_AvgPow}, we obtain  \eqref{Eqn:R_BF_LB_FBL_EH}. Proof of \eqref{Eqn:R_BF_UB_FBL_EH} is deferred to the Appendix D.
\end{proof}
\subsection{Discussion of Results}
\label{SubSec:Discuss_Results} 
As mentioned earlier, this is the first attempt at providing a finite blocklength analysis of a fading channel powered by an energy harvesting transmitter. We note that the second order coefficients, i.e., coefficients of $\sqrt{n}$, of lower and upper bounds of the new result do not match. In this regard, we note the unavailability of a matching second order coefficient  even for an AWGN channel with an energy harvesting transmitter  \cite{shenoy2016finite}, \cite{fong2018achievable}. 
\par Note that an energy harvesting wireless transmitter appropriately regulates the transmission power based on its CSI and energy availability.  However, in our proof of the inner bound, we decouple the channel coding from  the power control. Similarly, the energy harvesting mechanism  is also decoupled from the channel coding in  the Save and Transmit scheme. This decoupling suggests that our proof for the full CSIT case can be readily adapted to the case of no CSIT and full CSIR. Specifically, replacing   $G_1^2$ with $|H_1|^2/\sigma_N^2$ in the expressions in  \eqref{Eqn:Defn_2nd_OrdCoeff2_LB_EH} and using it in \eqref{Eqn:R_BF_LB_FBL_EH} yields a valid lower bound for the case of no CSIT and full CSIR. Similarly, the analysis in the Appendix holds for the no CSIT, full CSIR case as well. In fact, the analysis can be considerably simplified owing to the fact that the channel input does not depend on the fading states when CSIT is not available. Thus, replacing   $G_1^2$ with $|H_1|^2/\sigma_N^2$ to compute the terms in the upper bound in \eqref{Eqn:R_BF_UB_FBL_EH} results in a valid upper bound for the no CSIT, full CSIR case.
 \section{Moderate Deviation Bounds}
\label{sec_mod_dev}
In this section, we characterize the optimal rate of convergence, viz., the moderate deviation constant, of average probability of error, to zero when the channel coding rate is converging to the capacity at a rate between $\sqrt{n}$ and $n$. We derive the moderate deviation constant for the block fading channel in Section \ref{Sec:Model_Notation} with transmitters subjected to different kinds of power constraints mentioned therein. Let $p_{e,\text{avg}}^*(n,M_n)$ denote the minimum average probability of error over a given channel at codeword length $n$ and number of messages $M_n$. We begin with the following definition in \cite{polyanskiy2010channelmode}.
\begin{defn}
\label{Defn:Mod_Dev_Prop}
A channel with capacity $C$ is said to satisfy the moderate deviation property with constant $\mu$ if for any sequence of integers $M_n$ such that $\log M_n=nC-na_n$ where, the sequence $\{a_n,n \geq 1\}$ is strictly positive, $a_n\rightarrow 0$, $na_n^2\rightarrow \infty$, we have $\lim\limits_{n\rightarrow \infty}\frac{1}{na_n^2}\log  p_{e,\text{avg}}^*(n,M_n)=-\frac{1}{2\mu}$.
\end{defn}
We have the following result.
\begin{thm}
\label{Thm_mod_dev}
For the block fading channel in Section \ref{SubSec:Channel_Model} with transmitters subjected to \textbf{PP} and \textbf{AP}  constraints, the moderate deviation constant $\mu$ is such that $V_{\text{BF}}'(\bar{P})\leq -1/2\mu \leq V_{\text{BF}}(\bar{P})$, where $V_{\text{BF}}(\bar{P})$ is defined as \eqref{Eqn:Defn_V_BF_PP} and $V_{\text{BF}}'(\bar{P})\triangleq \mathbb{E}\left[\mathcal{L}\left(G_1^2 \right)\right]+n_c\mathbb{V} \left[C\left(G_1^2\right)\right]+\mathbb{V}\left[\mathcal{L}\left(G_1^2\right)\right]$. Also, under \textbf{EH} constraint, $\mu$ is such that $V_{EF}'\left(\overline{E}\right)\leq -1/2\mu \leq V_{EF}''\left(\overline{E}\right)$. 
\end{thm}

\begin{proof}
See Appendix \ref{App:Mod_Deviation}.
\end{proof}

\section{Numerical Examples}
\label{Sec:Numerical_Egs}
\subsection{Non-energy harvesting transmitter}
\label{SubSec:Num_Egs_Non_EH}
In this section, we compare numerically, the bounds obtained for the non-energy harvesting case (Theorem \ref{Th:PeakPow_Bnds} and Theorem \ref{Th:AvgPow_Bnds}). We will compare the bounds upto the second order  term in the rate expression. We assume fading distribution to be $\mathcal{CN}(0,\sigma_H^2)$.  Also, we assume $\sigma_N^2=4$ and $\epsilon=0.05$ and $n_c=10$. By fixing $\bar{P}=5$dB, we plot the convergence of various bounds to the channel capacity $\mathbf{C}(\bar{P})=0.6892$ bits/channel use (equation (\ref{Eqn:Cap_BlockFadeCha})), as $B$ increases, in Figure \ref{Fig:R_2Ord_vs_Pbar}. The acronyms LB, UB and TIC refer to the lower bound, the upper bound and truncated channel inversion, respectively. {We have also plotted the rate with no CSIT under $\mathbf{PP}$ constraint and with the rate corresponding to truncated channel inversion. Truncated channel inversion provides the delay limited capacity for point-to-point fading channels (\cite{hanly1998multiaccess}, Section III.B.1). We observe that power allocation that achieves the delay limited capacity achieves only an inferior rate to that achieved via the water filling power allocation scheme in the finite block length regime.

\begin{figure}
  
  \begin{minipage}[t]{0.45\textwidth}
   \begin{tikzpicture}[scale=0.85]
 
\begin{axis}[title={},xlabel={ Number of blocks, $B$},
    ylabel={Rate (bits per channel use)},xmin=900, xmax=1000,
    ymin=.1, ymax=1.5,
    xtick={900,920,940,960,980,1000},
    ytick={.1,.2,.3,.4,.5,.6,.7,.8,.9,1,1.2,1.3,1.5},
    legend pos=north east,
    xmajorgrids=true,
    ymajorgrids=true,
    grid style=dashed,
    shift={(10pt,0pt)}
]
   \addplot[color=ForestGreen,mark=diamond,thick]   
      table[x index=0,y index=6] {Rate_vs_Bv_nonEH.txt};
      \addlegendentry{\small Capacity}
      
        \addplot[color=Black,mark=Mercedes,thick]   
     table[x index=0,y index=5] {Rate_vs_Bv_nonEH.txt};
     \addlegendentry{\small Rate AP, UB PP}   
    
	 \addplot[color=RedViolet,mark=o,thick]   
     table[x index=0,y index=4] {Rate_vs_Bv_nonEH.txt};
     \addlegendentry{\small LB, PP}

    \addplot[color=MidnightBlue,mark=square,thick]          
      table[x index=0,y index=3] {Rate_vs_Bv_nonEH.txt};
      \addlegendentry{\small No CSIT }
      
      \addplot[color=Orange,mark=triangle,thick]          
      table[x index=0,y index=2] {Rate_vs_Bv_nonEH.txt};
      \addlegendentry{\small TIC, AP }
      \addplot[color=Purple,mark=star,thick]          
      table[x index=0,y index=1] {Rate_vs_Bv_nonEH.txt};
      \addlegendentry{\small TIC, PP }
    
 \end{axis}
\end{tikzpicture}
\caption{\small Comparison of rate versus $B$ for the non energy harvesting case. We fix $n_c=10$, $\bar{P}=5$dB, $\epsilon=0.05$, $\sigma_N^2=4$ and $\sigma_H^2=0.1$. }
\label{Fig:R_2Ord_vs_Pbar}
\end{minipage}
\hspace{1cm}
\begin{minipage}[t]{0.45\textwidth}
 \begin{tikzpicture}[scale=0.85]
 
\begin{axis}[title={},xlabel={ Power, $\bar{P}$},
    ylabel={Rate (bits per channel use)},xmin=1, xmax=15,
    ymin=.1, ymax=2.4,
    xtick={1,2,3,4,5,6,7,8,9,10,11,12,13,14,15},
    ytick={.1,.2,.4,.6,.8,1,1.2,1.4,1.5,1.6,1.7,1.9,2.1,2.4},
    legend pos=north west,
    xmajorgrids=true,
    ymajorgrids=true,
    grid style=dashed,
    shift={(10pt,0pt)}
]
   \addplot[color=ForestGreen,mark=diamond,thick]   
      table[x index=0,y index=6] {Rate_vs_Pbar.txt};
      \addlegendentry{\small Capacity}
    
	 \addplot[color=RedViolet,mark=o,thick]   
     table[x index=0,y index=5] {Rate_vs_Pbar.txt};
     \addlegendentry{\small Rate AP, UB PP}

      \addplot[color=Black,mark=Mercedes,thick]   
     table[x index=0,y index=4] {Rate_vs_Pbar.txt};
     \addlegendentry{\small LB, PP}    
     
    \addplot[color=MidnightBlue,mark=square,thick]          
      table[x index=0,y index=3] {Rate_vs_Pbar.txt};
      \addlegendentry{\small No CSIT}
      \addplot[color=Orange,mark=triangle,thick]          
      table[x index=0,y index=2] {Rate_vs_Pbar.txt};
      \addlegendentry{\small TIC, AP}
      \addplot[color=Purple,mark=star,thick]          
      table[x index=0,y index=1] {Rate_vs_Pbar.txt};
      \addlegendentry{\small TIC, PP}
    
 \end{axis}
\end{tikzpicture}
\caption{\small Comparison of rate versus $\bar{P}$ for the non energy harvesting case. We fix $n_c=10$, $B=1000$, $\epsilon=0.05$, $\sigma_N^2=4$ and $\sigma_H^2=0.1$.}
\label{Fig:R_2Ord_vs_B_nonEH}
\end{minipage}  
\end{figure} 
\begin{figure}
  \centering
 \begin{tikzpicture}[scale=0.85]
 
\begin{axis}[title={},xlabel={ Power, $\bar{P}$},
    ylabel={Rate (bits per channel use)},xmin=1, xmax=15,
    ymin=.1, ymax=4.3,
    xtick={1,2,3,4,5,6,7,8,9,10,11,12,13,14,15},
    ytick={.1,.6,1,1.4,1.7,2.1,2.4,2.8,3.2,3.6,4.3},
    legend pos=north west,
    xmajorgrids=true,
    ymajorgrids=true,
    grid style=dashed,
    shift={(10pt,0pt)}
]
   \addplot[color=ForestGreen,mark=diamond,thick]   
      table[x index=0,y index=6] {Rate_vs_Pbar_SH1SH2.txt};
      \addlegendentry{\small Rate AP, UB PP, $\sigma_H^2=0.4$}
    
	 \addplot[color=RedViolet,mark=o,thick]   
     table[x index=0,y index=5] {Rate_vs_Pbar_SH1SH2.txt};
     \addlegendentry{\small LB, PP, $\sigma_H^2=0.4$ }

      \addplot[color=Black,mark=Mercedes,thick]   
     table[x index=0,y index=4] {Rate_vs_Pbar_SH1SH2.txt};
     \addlegendentry{\small No CSIT, $\sigma_H^2=0.4$ }    
     
    \addplot[color=MidnightBlue,mark=square,thick]          
      table[x index=0,y index=3] {Rate_vs_Pbar_SH1SH2.txt};
      \addlegendentry{\small Rate AP, UB PP, $\sigma_H^2=0.1$}
      \addplot[color=Orange,mark=triangle,thick]          
      table[x index=0,y index=2] {Rate_vs_Pbar_SH1SH2.txt};
      \addlegendentry{\small LB, PP, $\sigma_H^2=0.1$ }
      \addplot[color=Black,mark=pentagon,thick]          
      table[x index=0,y index=1] {Rate_vs_Pbar_SH1SH2.txt};
      \addlegendentry{\small No CSIT, $\sigma_H^2=0.1$ }
    
 \end{axis}
\end{tikzpicture}

\caption{\small Comparison of rate for various values of $\sigma_H$ for non-energy harvesting case. We fix $n_c=10$, $B=1000$ $\epsilon=0.05$ and $\sigma_N^2=4$.}
\label{Fig:R_2Ord_vs_PSH1SH2_nonEH}
\end{figure}
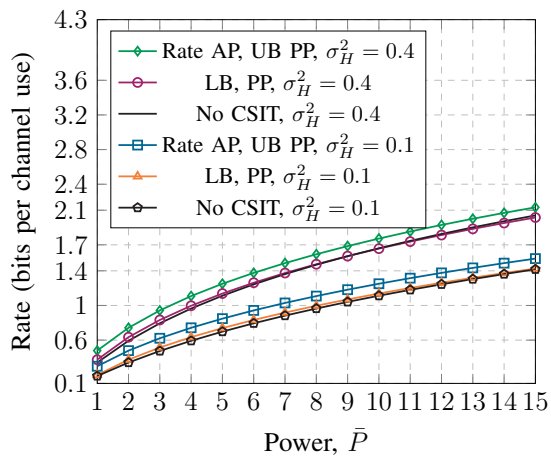  
\par In Figure \ref{Fig:R_2Ord_vs_B_nonEH}, we fix $n_c=10,~B=1000$ and provide  comparison of the various bounds for different values of $\bar{P}$.  By comparing the lower bound under the \textbf{PP} constraint with the rate for no CSIT case under \textbf{PP} constraint, we observe the rate enhancement possible due to the knowledge of CSIT (via power control). In Figure \ref{Fig:R_2Ord_vs_PSH1SH2_nonEH}, we compare the effect of fading parameter $\sigma_H$ on the second order approximation of rates. We observe that our second order approximation of LB under \textbf{PP} constraint provides better gains for lower values of $\sigma_H$ with other parameters fixed.

\subsection{Energy harvesting transmitter}
\label{SubSec:Num_Egs_EH}
In this  section, we compare our lower and upper bounds on the maximal coding rate, numerically. State-of-the-art ambient electric field energy harvesting technology promises a \emph{harvesting energy density} of $17$ micro Joules ($\mu \text{J}$) per unit time and unit volume (see, for instance \cite{akan2018internet} and  Table I therein). We consider $\overline{E}$  in the range $1$ to $17$ $\mu \text{J}$. For simplicity, we normalize and refer to $1 \mu \text{J}$ as a unit. We fix the average probability of error $\epsilon=0.1$ which is typically the targeted probability of error in \emph{massive machine-type transmission} \cite{popovski20185g}. Further, we fix $n_c=20,~B=400$, $\sigma_N^2=0.4$, $\sigma_H^2=0.9$ and $\sigma_E^2=0.1$. We assume fading distribution to be $\mathcal{CN}(0,\sigma_H^2)$. We plot the \emph{second order approximation of the lower and upper bounds} (terms excluding $o(\sqrt{n})$ term on the right hand side of \eqref{Eqn:R_BF_LB_FBL_EH} and \eqref{Eqn:R_BF_UB_FBL_EH} with an additional scaling of $1/n$)  as a function of $\overline{E}$ in Figure \ref{Fig:R_2Ord_vs_Ebar}. In obtaining the plot, we have optimized over the parameter $\alpha$ appearing in \eqref{Eqn:R_BF_LB_FBL_EH}, for $\alpha$ between $0.1~\text{and}~0.9$. For the chosen set of parameters, we observe that the second order approximation of the upper bound differs from that of the lower bound by a maximum of $\approx 15\%$. Also,  the duration of the saving phase $N_n$ (of the Save and Transmit scheme) in this example is approximately $846$ slots, the optimum $\alpha=0.104$ and $K_{\epsilon,\alpha}\approx 9.5 $. Next, we illustrate the convergence of the second order approximation of the bounds to the capacity as the number of blocks $B$ increases, in Figure \ref{Fig:R_2Ord_vs_B}. Here, we have chosen $\overline{E}=17$ units and other parameters are fixed as before. The  bounds in this case differ by $\approx 11\%$.  

 \begin{figure}
  \begin{minipage}[t]{0.45\textwidth}
 \begin{tikzpicture}[scale=0.85]
 
\begin{axis}[title={},xlabel={ Average harvested energy, $\overline{E}$},
    ylabel={Rate (bits per channel use)},xmin=1, xmax=17,
    ymin=1.7, ymax=5.7,
    xtick={1,3,5,7,9,11,13,15,17},
    ytick={1.7,2.3,2.9,3.5,4.1,4.7,5.3,5.7},
    legend pos=south east,
    xmajorgrids=true,
    ymajorgrids=true,
    grid style=dashed,
    shift={(10pt,0pt)}
]
   \addplot[color=ForestGreen,mark=diamond,thick]   
      table[x index=0,y index=3] {Rate_vs_Ebar.txt};
      \addlegendentry{Capacity}
    
	 \addplot[color=RedViolet,mark=o,thick]   
     table[x index=0,y index=2] {Rate_vs_Ebar.txt};
     \addlegendentry{Upper bound}    
    
      \addplot[color=MidnightBlue,mark=square,thick]          
      table[x index=0,y index=1] {Rate_vs_Ebar.txt};
      \addlegendentry{Lower bound}
    
 \end{axis}
\end{tikzpicture}

\caption{\small Comparison of rate versus $\overline{E}$ for the energy harvesting case. We fix $\epsilon=0.1$,  $n_c=20,~B=400$, $\sigma_N^2=0.4$, $\sigma_H^2=0.9$ and $\sigma_E^2=0.1$. }
\label{Fig:R_2Ord_vs_Ebar}
\end{minipage}  
\hspace{1cm}
\begin{minipage}[t]{0.45\textwidth}
 \begin{tikzpicture}[scale=0.85]
\begin{axis}[title={}, xlabel={ Number of blocks, $B$},ylabel={Rate (bits per channel use)}, xmin=450, xmax=550,ymin=4.5, ymax=5.7,xtick={450,470,490,510,530,550},
    ytick={4.5,4.6,4.7,4.8,4.9,5,5.1,5.2,5.3,5.4,5.5,5.6,5.7},
    legend pos= south east,
    xmajorgrids=true,
    ymajorgrids=true,
    grid style=dashed,
    shift={(10pt,0pt)}
]  
      \addplot[color=ForestGreen,mark=diamond,thick]   
      table[x index=0,y index=3] {Rate_vs_Bv.txt};
      \addlegendentry{Capacity}
    
	 \addplot[color=RedViolet,mark=o,thick]   
     table[x index=0,y index=2] {Rate_vs_Bv.txt};
     \addlegendentry{Upper bound}    
    
      \addplot[color=MidnightBlue,mark=square,thick]          
      table[x index=0,y index=1] {Rate_vs_Bv.txt};
      \addlegendentry{Lower bound}

 \end{axis}
\end{tikzpicture}

\caption{\small Comparison of rate versus $B$ for the energy harvesting case. We fix $\epsilon=0.1$,  $n_c=20,\overline{E}=17$, $\sigma_N^2=0.4$, $\sigma_H^2=0.9$ and $\sigma_E^2=0.1$.}
\label{Fig:R_2Ord_vs_B}
\end{minipage}
\end{figure} 

%
%
%
%
%
%
%
%

\section{Conclusion}
\label{Sec:Conclusion}
In this paper, we {have} obtained upper and lower bounds for the maximal coding rate over a block fading channel {with peak, average and energy harvesting} power constraints on the transmitted codeword. The bounds obtained shed light {on} the rate enhancement possible due to the availability of CSIT. The bounds also characterize the performance of water-filling power allocation in the finite block length regime. Further, we have derived moderate deviation bounds. These bounds characterizes the rate of convergence of probability of error to zero when the channel coding rate is allowed to scale with codeword length in a pre-specified way. 


\appendices
\section{Proof of Theorem \ref{Th:PeakPow_Bnds} }
\label{App:Proof_PeakPow}
 \subsection{Lower bound on rate with \textbf{PP} constraint}
 \label{SubSec:LB_PP}
\par  We prove the lower bound for a \emph{real} block fading channel with  gain $G_b=|H_b|\sqrt{\mathcal{P}_{\text{WF}}(|H_b|)}/\sigma_N$ in block $b$. The scheme presented can be readily adapted to the actual model in Section \ref{SubSec:Channel_Model} by \emph{independently coding over in-phase and quadrature channels} to get the lower bounds we use. Let $\mathbf{G}\triangleq (G_1,\hdots,G_B)$ and the additive noise in slot $[b,k]$ be $\tilde{Z}_{[b,k]}$, where $\tilde{Z}_{[b,k]}\sim \mathcal{N}(0,1/2)$.  In a block, the transmitter uses the channel $2n_c$ times ($n_c$ is defined as in Section \ref{SubSec:Channel_Model}). 
\subsubsection{Codebook Generation} 
\label{SubSubSec:CodebookGen_PP}
Given codeword length $n$, fix $\delta_n\in (0,1)$; for our choice of $\delta_n$, see Section \ref{SubSubSec:Analys_ConstrPowContr_PP}). Corresponding to each message $m\in[1:M]$, generate i.i.d., \textbf{PP} constrained codewords uniformly randomly from  $\mathcal{X}_{n}(1-\delta_n)\subset \mathbb{R}^n$ (where, as defined in Section \ref{Sec:Model_Notation}, $\mathcal{X}_{\ell}(q)$ denotes the surface of a sphere in $\mathbb{R}^\ell$ with radius $\sqrt{\ell q}$) and share it with the receiver. The codeword corresponding to message $m$, $\tilde{\mathbf{X}}(m)$ is such that  $\tilde{\mathbf{X}}(m)\perp \mathbf{G}$ for all $m$. For each $m$, the channel input symbol in slot $[b,k]$ is $\tilde{X}_{[b,k]}(m)$.  
\subsubsection{Power Control}
\label{SubSubSec:PowContr_PP}
\par To send message $m$, at the beginning of block $b$, if $G_b=0$ the transmitter does not transmit any symbols in that block and waits to observe $G_{b+1}$ (unless it is the last block). Else, if $G_b>0$, transmitter computes $|H_b|$ from $G_b$ and checks if 
\begin{equation}
\label{Eqn:Constr_Defn_PowContr_PP}
\sum_{\ell=1}^{b}\sum\limits_{k=1}^{2n_c}\tilde{X}_{[\ell,k]}^2(m)\mathcal{P}_{\text{WF}}\left(\vert H_\ell\vert\right) \leq 2Bn_c\bar{P}
\end{equation}
is met. If the constraint is not met, the \textbf{PP} constraint  in \eqref{Eqn:Defn_PCP_Constr} is violated and the  channel input is $0$ for blocks $\ell\in[b:B]$. That is, the transmission is halted. Else, if the constraint is met, the corresponding codeword symbols are transmitted in the block. If the transmitter sends codeword symbols successfully in all $B$ blocks without violating the constraint in \eqref{Eqn:Constr_Defn_PowContr_PP}, the constraint in \eqref{Eqn:Defn_PCP_Constr} is met. Then, the  channel output $\tilde{Y}_{[b,k]}=G_b\tilde{X}_{[b,k]}(m)+\tilde{Z}_{[b,k]}$, $b \in [1:B]$ and $k\in[1:2n_c]$.
\subsubsection{Decoding} 
\label{SubSubSec:Decoding_PP}
The decoder declares $\hat{S}$ as the transmitted message using the \emph{Neyman-Pearson decoder} employed in obtaining the $\beta\beta$ achievability bound in \cite{yang2018beta} (see Theorem 1 therein). We choose the test subject to a \emph{detection probability} $1-\epsilon_n'+\tau$, where we fix $\epsilon_n'\in(0,\epsilon)$  and $\tau\in(0,\epsilon_n')$ (our choices of $\epsilon_n'$ and $\tau$ are provided in Section \ref{SubSubSec:Analys_ConstrPowContr_PP}). Further, the \emph{auxiliary channel} in using the decoder is  $\mathbb{Q}_{\mathbf{\tilde{Y}}\vert \mathbf{G}}=\prod\limits_{b=1}^{B}\prod\limits_{k=1}^{2n_c}\mathcal{N}\left(0,\left(1-\delta_n\right)G_b^2\sigma_N^2+\sigma_N^2\right).$ Then, using the $\beta\beta$ achievability bound, under an average probability of error in decoding  $\epsilon_n'$, 
\begin{equation}
\label{Eqn:Defn_BetaBeta_Bound}
M \geq \frac{\beta_{\tau}\Big(\mathbb{P}_{\mathbf{\tilde{Y}},\mathbf{G}},\mathbb{Q}_{\mathbf{\tilde{Y}},\mathbf{G}}\Big)}{\beta_{1-\epsilon_n'+\tau}\Big(\mathbb{P}_{\mathbf{\tilde{X}}(1),\mathbf{G}}\mathbb{P}_{\mathbf{Y}\mid \mathbf{\tilde{X}}(1),\mathbf{G}},\mathbb{P}_{\mathbf{\tilde{X}}(1),\mathbf{G}}\mathbb{Q}_{\mathbf{\tilde{Y}}\mid \mathbf{G}}\Big)},
\end{equation}
where the $\beta$ notation is as in \cite{polyanskiy2010channel}.
\subsubsection{Analysis of $\beta\beta$ bound in \eqref{Eqn:Defn_BetaBeta_Bound}}
\label{SubSubSec:Analys_BetaBeta_PP}
\par For convenience, denote the inequality in \eqref{Eqn:Defn_BetaBeta_Bound} as $M\geq \beta_\tau\left(\mathbb{P}_1,\mathbb{Q}_1\right)/\beta_{1-\epsilon_n'+\tau}\left(\mathbb{P}_2,\mathbb{Q}_2\right).$  Here, $\mathbb{P}_1\equiv \mathbb{P}_{\mathbf{\tilde{Y}},\mathbf{G}}$ and other probability terms are similarly defined. To obtain a tractable lower bound for the RHS of \eqref{Eqn:Defn_BetaBeta_Bound} we upper bound $\beta_{1-\epsilon_n'+\tau}\left(\mathbb{P}_2,\mathbb{Q}_2\right)$ as
\begin{equation}
\label{Eqn:Ineq_UB_BetaP2Q2_PP}
\beta_{1-\epsilon_n'+\tau}\left(\mathbb{P}_2,\mathbb{Q}_2\right)\stackrel{(a)}{\le} \mathbb{Q}_2\left[ \frac{d\mathbb{P}_2}{d\mathbb{Q}_2} \ge  \gamma_0\right]\stackrel{(b)}{\le}\frac{c_2}{\gamma_0\sqrt{n}}.
\end{equation}
Here,  $\log \gamma_0=\mu_n+\sqrt{\nu_n}\Phi^{-1}\left(\epsilon_n''\right)$ and,  $\mu_n$ and $\nu_n$ are the mean and variance of $\log \frac{d\mathbb{P}_2}{d\mathbb{Q}_2}$ under $\mathbb{P}_2$, respectively. Also,  $\epsilon_n''=\epsilon_n'-\tau-({2c_1}/{\sqrt{n}})$ and $c_1>0$. In addition, $\gamma_0$ is such that $\mathbb{P}_2\left[\frac{d\mathbb{P}_2}{d\mathbb{Q}_2} \ge  \gamma_0\right] \ge 1-\epsilon_n'+\tau$. This follows from applying Berry-Esseen theorem (\cite{feller}, Chapter 16) to $\mathbb{P}_2[ \frac{d\mathbb{P}_2}{d\mathbb{Q}_2}\geq  \gamma_0\given \mathbf{\tilde{X}}(1)=\mathbf{\tilde{x}}(1) ]$ ($c_1$ is the Berry-Esseen correction term). Also, $(b)$ follows from Lemma 20 in \cite{polyanskiy2010channel1} and $c_2>0$. 
\par Next, from \cite{polyanskiy2011scalar} (see equation (59) therein), choosing $\tau$ to be a constant (i.e., not depending on $n$), $\beta_{\tau}\left(\mathbb{P}_1,\mathbb{Q}_1\right)=O\left(1\right)$. Combining this fact with \eqref{Eqn:Defn_BetaBeta_Bound} and \eqref{Eqn:Ineq_UB_BetaP2Q2_PP}, we obtain  
\begin{equation}
\label{Eqn:Ineq_LB1_Rate_PP}
\log M_p^*\geq\mu_n+\sqrt{\nu_n}\Phi^{-1}\left(\epsilon_n''\right)+O\left(\log n\right).
\end{equation}
\par By Taylor's theorem,  $ \mu_n \ge n\mathbf{C}(\bar{P})-n\delta_n'',$ where $\delta_n''=\big(\bar{P}/(2\lambda)\big)\delta_n+\delta_n^2$. Further, from Lemma \ref{Claim:V_BF_monotone}, $nV_{\text{BF}}(\bar{P})\geq \nu_n $. Combing these inequalities with the fact $\Phi^{-1}\left(\epsilon_n''\right)<0$  to lower bound the RHS of \eqref{Eqn:Ineq_LB1_Rate_PP}, we obtain $ R_p^*\geq \mathbf{C}(\bar{P})-\delta_n''+\sqrt{\frac{V_{\text{BF}}(\bar{P})}{n}}\Phi^{-1}\left(\epsilon_n''\right)+O\left(\frac{\log n}{n}\right).$
What remains is a characterization of $\delta_n''$ and $\epsilon_n''$ appearing in the above equation. Next, we do that by analysing the constraint in \eqref{Eqn:Constr_Defn_PowContr_PP}.
\label{SubSubSec:Analys_ConstrPowContr_PP}
\subsubsection{Analysis of the  constraint in  \eqref{Eqn:Constr_Defn_PowContr_PP}}
\label{SubSubSec:Analys_ConstrPowContr_PP}
\par Let $\mathcal{E}_T$ be the event that the constraint in \eqref{Eqn:Constr_Defn_PowContr_PP} is violated for some $b\in[1:B]$. Then,  $$\mathbb{P}\left[\mathcal{E}_T\cup \left\{S\neq\hat{S}\right\}\right]\leq \mathbb{P}\left[\mathcal{E}_T\right]+\epsilon_n'.$$ In order to upper bound $\mathbb{P}\left[\mathcal{E}_T\right]$, we choose $\delta_n$ (mentioned in Section \ref{SubSubSec:CodebookGen_PP}) to be $\left(2\lambda c_\epsilon\right)/\left(\bar{P}\sqrt{n}\right)$,
where for any $\alpha\in(0,1)$, $$c_\epsilon=\sqrt{{(2n_c+2)\log (1/\alpha \epsilon)}/{\log e}}.$$  Then, from Lemma \ref{Claim:P_E_T_UB_PP} we obtain the bound $\mathbb{P}\left[\mathcal{E}_T\right]\leq {\alpha\epsilon}+{c_3}/{\sqrt{n}},$ where $c_3=64(2n_c+3)^4$. Based on this bound, we choose  $\epsilon_n'=(1-\alpha)\epsilon-(c_3/\sqrt{n})$ so that 
$\mathbb{P}[\mathcal{E}_T\cup \{S\neq\hat{S})\}]\leq \epsilon.$ Here, we require $n$ such that $\epsilon_n'>0$. Further, as required in obtaining \eqref{Eqn:Ineq_LB1_Rate_PP}, we choose $\tau=(1-\alpha)\epsilon/2$, i.e., not depending on $n$. Thus, we obtain $\epsilon_n''$ in \eqref{Eqn:Ineq_LB1_Rate_PP} to be $(1-\alpha)\epsilon/2-c_4/\sqrt{n},$ where $c_4=2c_2+c_3$ and $n$ is such that $\epsilon_n''>0$. Also, by virtue of choice of $\delta_n$, $\delta_n''={c_\epsilon}/{\sqrt{n}}+(4\lambda^2c_\epsilon^2)/(\bar{P}^2n)$. Plugging in the expressions for $\epsilon_n''$ and $\delta_n''$ in the lower bound on $R_p^*$, denoting $\epsilon_\alpha=(1-\alpha)\epsilon/2$ and applying Taylor's theorem, we obtain the required result.
  \qed
  \begin{lem}
\label{Claim:V_BF_monotone}
With the notation as in Section \ref{Sec:Model_Notation} and $\alpha\in[0,1]$, let $V_{\text{BF}}(\bar{P},\alpha)\triangleq \mathbb{E}\left[V\left(\alpha G_1^2 \right)\right]+n_c\mathbb{V} \left[C\left(\alpha G_1^2\right)\right]+\mathbb{V}\left[\mathcal{L}\left(\alpha G_1^2\right)\right].$ The function $V_{\text{BF}}(\bar{P},\alpha)$ is monotonically decreasing in $\alpha$.
\end{lem}
 \begin{proof}
Observe that $V_{\text{BF}}(\bar{P},\alpha)$ can be written as $\left[1-\mathbb{E}^2\left[1-\mathcal{L}(\alpha G_1^2)\right]\right]+n_c\mathbb{V}\left[C\left(\alpha G_1^2\right)\right].$ Since $\mathcal{L}(\alpha G_1^2)$ is monotonically increasing in $\alpha$, $$\left[1-\mathbb{E}^2\left[1-\mathcal{L}(G_1^2)\right]\right]\geq \left[1-\mathbb{E}^2\left[1-\mathcal{L}(\alpha G_1^2)\right]\right].$$ Thus, the result follows if we show $\mathbb{V}[C(G_1^2)]\geq \mathbb{V}[C(\alpha  G_1^2)]$. Towards that, first we invoke Taylor's theorem and express $$C\big(G_1^2\big)=C(\alpha G_1^2)+(1-\alpha)({\mathcal{L}(uG_1^2)}/{u}),$$ for some $u\in(\alpha,1).$
 Next, for convenience we denote  $f_1(G_1)=C(\alpha G_1^2)$, $\theta=(1-\alpha)/u$ and $f_2(G_1)=\theta\mathcal{L}(uG_1^2)$ and define $\gamma=\mathbb{E}\left[f_1(G_1)f_2(G_1)\right]- \mathbb{E}\left[f_1(G_1)\right]\mathbb{E}\left[f_2(G_1)\right].$ Then, using Chebyshev's association inequality (\cite{boucheron2013concentration}, Theorem 2.14), we observe $\gamma\geq 0$. Using this observation and the Taylor's expansion for $C(G_1^2)$, $$\mathbb{V}[C(G_1^2)]=\mathbb{V}\left[f_1(G_1)\right]+\mathbb{V}\left[f_2(G_1)\right]+2\gamma \geq \mathbb{V}\big[C(\alpha G_1^2)\big].$$
Combining this with the monotonicity of $\mathcal{L}(\alpha G_1^2)$ in $\alpha$, we obtain the required result.
 \end{proof}

\begin{lem}
\label{Claim:P_E_T_UB_PP}
Fix the channel described in Section \ref{SubSec:Channel_Model}. With $\mathcal{X}_n(\cdot)$ defined as in Section \ref{Sec:Model_Notation}, let $\mathbf{\tilde{X}}$ be uniformly distributed on $\mathcal{X}_n\left(1-\delta_n\right)\subset \mathbb{R}^{n}$, for $\delta_n>0$. Let  $\mathcal{P}_{\text{WF}}(\cdot)$ be the water filling power allocation function with average power $\bar{P}$, as in \eqref{Eqn:Cap_BlockFadeCha}.  Further, the event $\mathcal{E}_T$ corresponding to violating the transmission constraint in \eqref{Eqn:Constr_Defn_PowContr_PP} in some block is $$\mathcal{E}_T = \bigcup_{b=1}^{B}\left\{ \sum\limits_{\ell=1}^{b}||\mathbf{\tilde{X}}_\ell||^2\mathcal{P}_{\text{WF}}\left(|H_\ell|\right) > n\bar{P}\right\},$$ where the notation $\mathbf{\tilde{X}}_\ell$ is as mentioned in Section \ref{Sec:Model_Notation}. Then, with $\lambda$  as in the definition of $\mathcal{P}_{\text{WF}}(\cdot)$,
$$\mathbb{P}\left[\mathcal{E}_T\right]\leq \exp\left(-\frac{n\bar{P}^2\delta_n^2}{8(n_c+1)\lambda^2}\right) +\frac{64(2n_c+3)^4}{\sqrt{n}}.$$
\end{lem}
\begin{proof}
Let $U_b\triangleq	||\mathbf{\tilde{x}}_b||^2\mathcal{P}_{\text{WF}}\left(|H_\ell|\right)-||\mathbf{\tilde{x}}_b||^2\bar{P}$. Making use of the notation for summation $\mathbf{S}_B(\cdot)$ defined in Section \ref{Sec:Model_Notation}, observe $\mathbb{P}\left[\mathcal{E}_T\Bigiven \mathbf{\tilde{X}}=\mathbf{\tilde{x}}\right]=\mathbb{P}[\mathbf{S}_B(U_1) > n\delta_n\bar{P}].$ Since $\mathcal{P}_{\text{WF}}\left(\cdot\right)$ is such that  $0\leq \mathcal{P}_{\text{WF}}\left(\cdot\right)\leq \lambda$, $U_b$ is a zero mean, bounded random variable whose support set is contained in $\left[-||\mathbf{\tilde{x}}_b||^2\lambda,||\mathbf{\tilde{x}}_b||^2\lambda\right]$. From Hoeffding's inequality (\cite{boucheron2013concentration}, Theorem 2.8), $$\mathbb{P}\left[\mathcal{E}_T \Bigiven \mathbf{\tilde{X}}=\mathbf{\tilde{x}}\right]\leq \exp\left(-{n^2\tilde{\delta_n}^2}/{\mathbf{S}_B(||\mathbf{\tilde{x}}_1||^4)}\right),$$ where $\tilde{\delta}_n=\big((\delta_n\bar{P})/(2\lambda)\big)$. The above inequality yields $$\mathbb{P}\left[\mathcal{E}_T \right]\leq \mathbb{E}\left[\exp\left(-{n^2\tilde{\delta_n}^2}/{S_B(||\mathbf{\tilde{X}}_1||^4)}\right)\right].$$
\par To further bound $\mathbb{P}\left[\mathcal{E}_T \right]$, let $\mathbf{\hat{X}}$ be an i.i.d. $\mathcal{N}(0,1)$ vector. Then, $\mathbf{\tilde{X}}\stackrel{D}{=}\big(\sqrt{n(1-\delta_n)}/||\mathbf{\hat{X}}||\big)\cdot\mathbf{\hat{X}}$. From Chebyshev's inequality $\mathbb{P}\left[||\mathbf{ \hat{X}}||^2 > (n+n^{3/4})\right]\leq (2/\sqrt{n})$,  \begin{equation}
\label{Eqn:Ineq_UB_P_E_T_Claim_PP1}
\mathbb{P}\left[\mathcal{E}_T \right]\leq \mathbb{E}\left[\exp\left(-{n^2\hat{\delta}_n^2}/{\mathbf{S}_B(||\mathbf{\hat{X}}_1||^4)}\right)\right]+{2}/{\sqrt{n}},
\end{equation}
where we denote $\hat{\delta}_n=\big(\tilde{\delta}_n/(1-\delta_n)\big)\left(1-(1/n^{1/4})\right)$. Next, we tackle the expectation term on the RHS of the above equation. Towards that, we apply Chebyshev's inequality to obtain 
\begin{equation}
\label{Eqn:Ineq_UB_P_E_T_Claim_PP3}
\mathbb{P}\left[\mathbf{S}_B\left(||\mathbf{\hat{X}}_1||^4\right) > B\mathbb{E}\left[||\mathbf{\hat{X}}_1||^4\right]+B^{3/4}\right]\leq \frac{32(2n_c+3)^4}{\sqrt{n}}.
\end{equation}
Combining \eqref{Eqn:Ineq_UB_P_E_T_Claim_PP1} and \eqref{Eqn:Ineq_UB_P_E_T_Claim_PP3} with the fact  $\mathbb{E}[||\mathbf{\hat{X}}_1||^4]= 4(n_c+1)n_c$, we obtain the result.
 \end{proof}
 \subsection{Upper bound on rate with \textbf{PP} constraint}
 \label{SubSec:UB_PP}
\par Without loss of optimality (upto second order), we assume that \textbf{PP} constraint in \eqref{Eqn:Defn_PCP_Constr} is satisfied with equality (see for instance, \cite{polyanskiy2010channel1}, Lemma 65). Also, we assume that CSIT is known \emph{non-causally} to the transmitter, as it can only improve the rate. Let $M$ denote the size of an arbitrary codebook $\mathcal{C}$ with codewords of length $n$ satisfying the \textbf{PP} with equality, and average probability of error $\epsilon$ if used over the block fading channel described in Section \ref{Sec:Model_Notation}.  We divide the proof into various steps. Next, we explain those steps. 
\subsubsection{Choice of auxiliary channel} 
\label{SubSubec:Ch_Aux_PP_UB}
In order to invoke the meta converse (\cite{polyanskiy2010channel},  Theorem 26), we choose the auxiliary channel $\mathbb{Q}_{\mathbf{Y}\big\vert\mathbf{H}}=\prod_{b=1}^{B}\prod_{k=1}^{2n_c}\mathcal{C}\mathcal{N}\Big(0,|H_b|^2\mathcal{P}_{\text{WF}}(|H_b|)+\sigma_N^2\Big).$
 Note that, in $[b,k]\textsuperscript{th}$ slot, output $Y_{[bk]}$ depends only on $H_b$  and is independent of the channel input. 
 \subsubsection{Meta-converse bound and its relaxation} 
 \label{SubSubec:UB_via_MetaConv_PP_UB}
Let $\beta$ notation be as in \cite{polyanskiy2010channel}, $\mathbb{P}_{\mathbf{Y}\vert\mathbf{H},\mathbf{X}}$ be the channel in Section \ref{SubSec:Channel_Model} and $\epsilon'$ be the average probability of error if $\mathcal{C}$ is used over $\mathbb{Q}_{\mathbf{Y}\vert\mathbf{H},}$. Invoking meta-converse \cite{polyanskiy2010channel} and Claim \ref{Claim:MetaConv_Fade_SpecAuxCha},$\hspace{20pt}
\beta_{1-\epsilon}(\mathbb{P}_{\mathbf{X},\mathbf{H},\mathbf{Y}},\mathbb{Q}_{\mathbf{X},\mathbf{H},\mathbf{Y}})\leq 1/M.$
We lower bound $ \beta_{1-\epsilon}(\mathbb{P}_{\mathbf{X},\mathbf{H},\mathbf{Y}},\mathbb{Q}_{\mathbf{X},\mathbf{H},\mathbf{Y}})$ (see \cite{polyanskiy2010channel}, equation (102)) as $$\left(\left(\mathbb{P}\left[ \mathcal{I}_\gamma\right]-\epsilon\right)^+\big/\gamma\right)\leq  \beta_{1-\epsilon}(\mathbb{P}_{\mathbf{X},\mathbf{H},\mathbf{Y}},\mathbb{Q}_{\mathbf{X},\mathbf{H},\mathbf{Y}}),$$ where $\gamma$ is an arbitrary positive real number, $\mathcal{I}_\gamma \triangleq \Big\{i(\mathbf{X},\mathbf{H},\mathbf{Y})\leq \log \gamma\Big\}$ 
and $i(\mathbf{X},\mathbf{H},\mathbf{Y}) \triangleq \log \frac{d\mathbb{P}_{\mathbf{X},\mathbf{H}, \mathbf{Y}}}{d\mathbb{Q}_{\mathbf{X},\mathbf{H},\mathbf{Y}}}$.
Considering the relaxation of meta-converse, and optimizing over the set $\mathcal{F}_n^{(p)}$ of input distributions with support on $\mathcal{X}_n(\bar{P})$, $$ \log M_p^* \leq \log \gamma -\inf\limits_{\mathcal{F}_n^{(p)}}\log\Big(\mathbb{P}\big[\mathcal{I}_\gamma \big]-\epsilon\Big)^+.$$
\subsubsection{Lower bounding $\mathbb{P}[\mathcal{I}_\gamma]$ and the final bound}
\label{SubSubec:LB_P_I_gamma_PP_UB}
From Lemma \ref{Lem:P_I_gamma_LB_PP_UB}, for some positive constants $c_{4},~c_{6}$ depending only on the fading distribution and $\sigma_N^2$ and an appropriate choice of $\gamma$, $\mathbb{P}\left[\mathcal{I}_\gamma\right]\geq \epsilon+{c_6}/{n^{c_4}}$. 
Plugging this in the bound on $\log M_p^*$, we obtain $
\log M_p^* \leq \log \gamma +c_4\log n-\log c_4.$ With $V_{\text{BF}}(\bar{P})$ as in \eqref{Eqn:Defn_V_BF_PP}, choose $$\log \gamma=n\mathbf{C}(\bar{P})+\Phi^{-1}(\epsilon)\sqrt{nV_{\text{BF}}(\bar{P})}+o(\sqrt{n}).$$ Plugging this in the bound on $\log M_p^*$, we obtain the final bound in \eqref{Eqn:Bnd_UB_R_p_star_Th_PP}.
 \qed 
 
 \begin{claim}
 \label{Claim:MetaConv_Fade_SpecAuxCha}
For the auxiliary channel defined in Section \ref{SubSubec:Ch_Aux_PP_UB}), given an arbitrary codebook  $\mathcal{C}$ of size $M$, the average probability of decoding error $\epsilon'=1/M$.
 \end{claim}
 \begin{proof}
Let $S$ be the message, $\hat{S}$ the output of the decoder and $\mathcal{S}\triangleq \mathbb{C}^B \times \mathcal{X}_n(\bar{P}) \times  \mathbb{C}^n $. Then
\begin{align*}
\epsilon'&\stackrel{(a)}{=}\sum_{m=1}^{M}\int\limits_{\mathcal{S}}\frac{1}{M}dF_{\mathbf{H}}\mathbb{Q}\big[\mathbf{y}|\mathbf{h}\big]\mathbb{P}\big[\hat{S}=m|\mathbf{h},\mathbf{y}\big]d\mathbf{y}dF_{\mathbf{X}|\mathbf{h}}=\frac{1}{M},
\end{align*} 
where $(a)$ follows by noting $S\perp\mathbf{H}$, definition of $\mathbb{Q}$ channel and the Markov relation  $S-(\mathbf{H},\mathbf{Y})-\hat{S}$.
 \end{proof}

\begin{claim}
\label{Claim:L_bound}
For $h\in \mathbb{C}$ , $\lambda$ and $\mathcal{P}_{\text{WF}}(\cdot)$ as defined in \eqref{Eqn:Cap_BlockFadeCha}, 
${|h|^2}\big/{\left(\sigma_N^2+|h|^2\mathcal{P}_{\text{WF}}(|h|)\right)}\leq{1}/{\lambda}.$ Also, with $\mathcal{L}(\cdot)$ as in Section \ref{Sec:Model_Notation},  $\mathcal{L}\big(|h|^2\mathcal{P}_{\text{WF}}(|h|)\big)=\frac{\mathcal{P}_{\text{WF}}(|h|)}{\lambda}.$
\end{claim}
\begin{proof}
Since $\mathcal{P}_{\text{WF}}(|h|)\geq \Big(\lambda-\frac{\sigma_N^2}{|h|^2}\Big)$, we have $\lambda |h|^2 \leq \sigma_N^2 +|h|^2 \mathcal{P}_{\text{WF}}(|h|)$ and the first claim follows. Now, let $\mathcal{A} \triangleq \{\lambda |h|^2 > \sigma_N^2\}$. Observe $\mathcal{P}_{\text{WF}}(|h|)= \mathbbm{1}_{\mathcal{A}} \mathcal{P}_{\text{WF}}(|h|)$. Hence, $$\mathcal{L}\big(|h|^2\mathcal{P}_{\text{WF}}(|h|)\big)=\mathbbm{1}_{\mathcal{A}}{|h|^2\mathcal{P}_{\text{WF}}(h)}\big/{\left(\sigma_N^2+|h|^2\mathcal{P}_{\text{WF}}(|h|)\right)}.$$ However, $\mathbbm{1}_{\mathcal{A}}\frac{|h|^2\mathcal{P}_{\text{WF}}(h)}{\sigma_N^2+|h|^2\mathcal{P}_{\text{WF}}(|h|)}=\mathbbm{1}_{\mathcal{A}}\frac{|h|^2\mathcal{P}_{\text{WF}}(|h|)}{\sigma_N^2+|h|^2\left(\lambda-\left({\sigma_N^2}/{|h|^2}\right)\right)}$ which, is equal to ${\mathcal{P}_{\text{WF}}(|h|)}/{\lambda }.$
\end{proof}

 \begin{lem}
\label{Lem:P_I_gamma_LB_PP_UB}
For $\gamma>0$, let $\mathcal{I}_\gamma$ be defined as in Section \ref{SubSubec:UB_via_MetaConv_PP_UB}, $\epsilon$ be the average probability of error for the channel model in Section \ref{SubSec:Channel_Model}. There exist positive constants $c_{4},~c_{6}$ depending only on the channel parameters such that  $\mathbb{P}\big[\mathcal{I}_\gamma\big]\geq  \epsilon+{c_{6}}/{n^{c_{4}}}.$
\end{lem}
\begin{proof}
With $\mathbf{Y}$ being the output vector corresponding to the original channel $\mathbb{P}_{\mathbf{Y|X',H}}$, $\mathcal{I}_\gamma$ in Section \ref{SubSubec:UB_via_MetaConv_PP_UB} can be equivalently written as $\mathcal{I}_\gamma=\left\{\log\frac{d\mathbb{P}_{\mathbf{Y|X',H}}}{d\mathbb{Q}_\mathbf{Y|H}} \leq \log \gamma\right\}$. To analyse $\mathbb{P}[\mathcal{I}_\gamma]$, first we consider a random variable with the same distribution as the log likelihood ratio in the definition of $\mathcal{I}_\gamma$. Towards that, we define $$L_b^{(1)}= \sum_{i=1}^{2}\frac{H_{b,i}^2||\mathbf{X'}_{b,i}||^2}{\sigma_N^2\left(1+G_b^2\right)},~L_b^{(2)}=\sum_{i=1}^{2}\frac{H_{b,i}\langle\mathbf{X'}_{b,i},\mathbf{Z}_{b,i}\rangle}{\sigma_N^2\left(1+G_b^2\right)},$$  $$L_b^{(3)}=-\sum_{i=1}^{2}\frac{G_b^2||\mathbf{Z}_{b,i}||^2}{\sigma_N^4\left(1+G_b^2\right)},~L_b= {n_c}C(G_b^2)+\sum_{\ell=1}^3L_b^{(\ell)}.$$   Using the notation from Section \ref{Sec:Model_Notation}, define $\mathbf{S}_B(L_1)=\sum_{b=1}^BL_b$ and $\mathcal{I}_\gamma'=\{\mathbf{S}_B(L_1)\leq \log \gamma_n\}$.  Then, by direct verification, we have $\mathbf{S}_B(L_1)\stackrel{(D)}{=}\log\frac{d\mathbb{P}_{\mathbf{Y|X',H}}}{d\mathbb{Q}_\mathbf{Y|H}},$
Using this, we obtain $\mathbb{P}[\mathcal{I}_\gamma]=\mathbb{P}[\mathcal{I}_\gamma']$.
  Next, we lower bound $\mathbb{P}[\mathcal{I}_\gamma']$.  Define $$ L_b^{(4)}=\sum_{i=1}^{2}{\left(n_c\bar{P}\sigma_N^2-\mathcal{P}_{\text{WF}}(|H_b|)||\mathbf{Z}_{b,i}||^2\right)}/{\lambda \sigma_N^2 }.$$
   Denoting  $L_b' =  {n_c} C(G_b^2)+L_b^{(2)}+L_b^{(4)}$ and  $\mathbf{S}_B(L_1')=\sum_{b=1}^BL_b'$, define $\mathcal{I}_\gamma''=\{\mathbf{S}_B(L_1')\leq \log \gamma\}$. Using Claim \ref{Claim:L_bound} in Appendix \ref{App:Proof_PeakPow}, we obtain $\mathbf{S}_B(L_1)\leq \mathbf{S}_B(L_1')$. From this inequality, we have
$\mathbb{P}[\mathcal{I}_\gamma']\geq\mathbb{P}[\mathcal{I}_\gamma''].
$  
 \par Next, we intersect the event $\mathcal{I}_\gamma''$ with certain \emph{high probability} events so as to further  lower bound $\mathbb{
 P}[\mathcal{I}_\gamma'']$ appropriately. Towards defining the events, consider mutually independent collection of  i.i.d. $\mathcal{N}(0,1)$ random variables $\{U_b^{(i)},~b\in[1:B]\}$, for  $i\in\{1,2\}$. Also, recall the definition of the notation $C(\mathbf{G}^2)$ from Section \ref{Sec:Model_Notation} and denote $W_B^{(1)}=\left(||C(\mathbf{G}^2)||_1-Bn_c\mathbf{C}(\bar{P})\right)/ \sqrt{\mathbb{V}\left[n_cC(G_1^2)\right]}.$ Then, for some positive constant $c_1$ (depending on the parameters fixed in this analysis), define the event $
\mathcal{E}_2= \left\{W_B^{(1)}\leq \sum_{b=1}^BU_b^{(1)}+2c_1\log B \right \}.$ 
  Further, let $W_B^{(2)}=\sum_{b=1}^B(L_b^{(4)}/\sqrt{\mathbb{V}[L_1^{(4)}]}).$ Then, for $c_2>0$ (depending on the other parameters), let
$$\mathcal{E}_3=\left \{W_B^{(2)}\leq \sum_{b=1}^BU_b^{(2)}+2c_2\log B\right \}.$$
   
 In order to show that the event $\mathcal{E}_2 \cap \mathcal{E}_3$ is indeed a high probability event, we make use of the strong approximation principle of partial sum of i.i.d. random variables. Towards that, since $\mathbb{E}\left[|H_1|^2\right]<\infty$, observe that there exists $t_1>0$ such that $\mathbb{E}[\exp(tC(G_1^2))]<\infty$, for  $|t|\leq t_1$. Similarly, since $\mathcal{P}_{\text{WF}}(\cdot)\leq \lambda$, there exists $t_2>0$ such that  $\mathbb{E}[\exp(tL_1^{(4)})]<\infty$, for  $|t|\leq t_2$. Hence, using strong approximation of partial sum of i.i.d. random variables (\cite{csorgo2014strong}, Theorem 2.6.2), there exist positive constants $c_3,~c_4$ such that $\mathbb{P}\left[\mathcal{E}_2^c\cup \mathcal{E}_3^c\right]\leq {c_3}/{B^{c_4}}.$
\par Now, denote $\nu_1=\mathbb{V}\left[n_cC(G_1^2)\right]+\mathbb{V}[L_1^{(4)}].$ Also, let $$\nu_2\left(\mathbf{X'},\mathbf{H}\right)\equiv\nu_2=\frac{1}{B}\sum\limits_{b=1}^B\sum\limits_{i=1}^{2}\frac{2H_{b,i}^2||\mathbf{X'}_{b,i}||^2}{\sigma_N^2\left(1+G_b^2\right)^2}.$$ Finally, let $c_5=\max\{c_1,c_2\}$, where $c_1$ and $c_2$ are as in the definitions of $\mathcal{E}_2$ and $\mathcal{E}_3$.  Then, using the definition of events $\mathcal{E}_1$ and $\mathcal{E}_2$, the strong approximation bound mentioned above, and the fact that sum of independent standard normal random variables is a Gaussian random variable with appropriate mean and variance, we lower bound $\mathbb{P}\left[\mathcal{I}_\gamma''\right]$ as
\begin{equation}
\label{Eqn:Bnd_PI_gam_dprime_LB_SA_Lem_PP_UB}
\mathbb{P}\left[\mathcal{I}_\gamma''\right]\geq \mathbb{E}\left[\Phi\left(\frac{\log \gamma-Bn_c\mathbf{C}(\bar{P})-c_5\log B}{\sqrt{B\left(\nu_1+\nu_2\left(\mathbf{X'},\mathbf{H}\right)\right)}}\right)\right]- \frac{c_3}{B^{c_4}}.
\end{equation}	
 
Next, conditioned on $\mathbf{H}=\mathbf{h},~\mathbf{X'}=\mathbf{x'}$, we further lower bound $\nu_2$ appearing on the RHS of \eqref{Eqn:Bnd_PI_gam_dprime_LB_SA_Lem_PP_UB}  by optimizing over choice of $\mathbf{x'}\in\mathcal{X}_n(\bar{P})$. Towards that, we note that 
\begin{equation}
\label{Eqn:Opt_equality_Lem_PP_UB}
\min_{\mathbf{x'}\in\mathcal{X}_n(\bar{P})}\nu_2=\min_{\mathscr{P}\in\mathcal{P}_B(\bar{P})}\mathbb{E}_{V_1,V_2}^B\left[\sum\limits_{i=1}^{2}\frac{2V_{i}^2\mathscr{P}(V_i)}{\sigma_N^2\left(1+f(V_1,V_2)^2\right)}\right],
\end{equation}
where $V_1,V_2$ are independent random variables with common distribution being the \emph{empirical distribution} of $\mathbf{h}$, $\mathbb{E}_{V_1,V_2}^B\left[\cdot\right]$ denotes the expectation with respect to the \emph{empirical distribution} of $\mathbf{h}$, $f(V_1,V_2)=(V_1^2+V_2^2)\sqrt{\mathcal{P}_{\text{WF}}(V_1^2+V_2^2)}/\sigma_N$   and $$\mathcal{P}_B(\bar{P})\triangleq \{\mathscr{P}:\mathbb{R}\mapsto \mathbb{R}_+,\mathbb{E}_{V_1}^B[\mathscr{P}(V_1)]=\bar{P}]\}.$$ Next, we note that $\mathscr{P}$ that attains the above minimum has the property $\mathscr{P}^*(h)\leq \lambda,$ for $h\in\{h_{b,i},b\in[1:B],i=1,2\}$. This is justified by noting that for $\mathscr{P}^*$ that attains the above lower bound, replacing $\mathscr{P}^*$ with  $\tilde{\mathscr{P}}(h)=\mathscr{P}^*(h)$ if  $h> 1/\sqrt{\lambda},~ \mathcal{P}_{\text{WF}}(h)\geq \mathscr{P}^*(h)$, and  $\tilde{\mathscr{P}}(h)=\mathcal{P}_{\text{WF}}(h)$, else. With this,we further lower bound $\nu_2$ if we choose $\log \gamma<Bn_c\mathbf{C}(\bar{P})+c_5\log B.$ We will ensure that the choice of $\gamma$ is subject to this constraint. Using $\mathscr{P}^*(h)\leq \lambda$, next we apply McDiarmid's inequality (\cite{boucheron2013concentration}, Theorem 6.2) to replace $\nu_2$ with $\mathbb{E}[\nu_2]-\delta_B'$, for some $\delta_B'=o(\sqrt{B})$ to further lower bound the RHS of \eqref{Eqn:Bnd_PI_gam_dprime_LB_SA_Lem_PP_UB}. Thus, we obtain $\mathbb{P}\left[\mathcal{I}_\gamma''\right]\geq \epsilon + c_6/n^{c_4}$ 
by choosing $\log \gamma = Bn_c\mathbf{C}(\bar{P})+\sqrt{V_\text{BF}(\bar{P})}\Phi^{-1}(\epsilon+\sqrt{n})+c_5\log B$ and $c_6>0$. Here, $V_\text{BF}(\bar{P})$ is as in \eqref{Eqn:Defn_V_BF_PP} and is obtained as a lower bound to $\nu_1+\mathbb{E}[\nu_2]$. Also, condition on $\gamma$ is ensured as $\epsilon<1/2$. This concludes the proof. 

 \end{proof}

 \section{Proof of Theorem \ref{Th:AvgPow_Bnds} }
\label{App:Proof_AvgPow}
\subsection{Lower bound on rate with \textbf{AP} constraint}
 \label{SubSec:LB_AP}
\par The proof is similar to that of the lower bound with \textbf{PP} constraint in Section \ref{SubSec:LB_PP}. With the same setting as therein and decoupling coding and  power control policy, it is immediate to verify that $\mathbb{E}_{\mathbf{H}}\left[\sum_{b=1}^{B}\sum\limits_{k=1}^{2n_c}\tilde{X}_{[b,k]}^2(m)\mathcal{P}_{\text{WF}}\left(\vert H_b\vert\right)\right] = 2Bn_c\bar{P}$
is satisfied for all $m$. Hence, using the same analysis as in the \textbf{PP} case (excluding the analysis of \textbf{PP} constraint violation event), $R_a^*\geq \mathbf{C}(\bar{P})+\sqrt{\frac{V_{\text{BF}}(\bar{P})}{n}}\Phi^{-1}\left(\epsilon\right)+O\left(\frac{\log n}{n}\right).$ 
 \qed
 \subsection*{Upper bound on rate with \textbf{AP} constraint }
\par The proof follows exactly along the same lines of the proof in Section \ref{SubSec:UB_PP} till lower bounding the event $\mathbb{P}[\mathcal{I}_\gamma]$ in Section \ref{SubSubec:LB_P_I_gamma_PP_UB}. In case of \textbf{AP} case, we make use of Lemma \ref{Lem:P_I_gamma_LB_AP_UB} to lower bound this probability term.  From Lemma \ref{Lem:P_I_gamma_LB_AP_UB}, for some positive constants $c_{9},~c_{10}$ and an appropriate choice of $\gamma$, $\mathbb{P}\left[\mathcal{I}_\gamma\right]\geq \epsilon+{c_9}/{n^{c_{10}}}.$ Using this estimate on the relaxed meta converse as in \textbf{PP} case, $\log M_a^* \leq \log \gamma +c_{10}\log n-\log c_9.$ With $V_{\text{BF}}(\bar{P})$ as in \eqref{Eqn:Defn_V_BF_PP}, choose $\log \gamma=n\mathbf{C}(\bar{P})+\Phi^{-1}(\epsilon)\sqrt{nV_{\text{BF}}(\bar{P})}+o(\sqrt{n})$. Plugging this in the above bound, $R_a^* \leq \mathbf{C}(\bar{P})+\Phi^{-1}(\epsilon)\sqrt{\frac{V_{\text{BF}}(\bar{P})}{n}}+o\left(\frac{1}{\sqrt{n}}\right).$
 \qed 
  \begin{lem}
\label{Lem:P_I_gamma_LB_AP_UB}
Let $\mathcal{I}_\gamma$ be defined as in Section \ref{SubSubec:UB_via_MetaConv_PP_UB} and $\epsilon$ be the average probability of error for the channel in Section \eqref{SubSec:Channel_Model}. There exists positive constants $c_9>0$ and $c_{10}>0$ such that  $\mathbb{P}\big[\mathcal{I}_\gamma\big]\geq  \epsilon+c_9/n^{c_{10}}.$
\end{lem}
\begin{proof}
In proving the lemma, we make use of the method and notation in Lemma \ref{Lem:P_I_gamma_LB_PP_UB}. With \textbf{AP} constraint as well, using the arguments therein, it is easy to observe $\mathbb{P}[\mathcal{I}_\gamma]=\mathbb{P}[\mathcal{I}_\gamma'].$
  Next, to lower bound $\mathbb{P}[\mathcal{I}_\gamma']$, define $
  L_b^{(5)}=\sum_{i=1}^{2}\frac{n_c\mathcal{P}(H_b)\sigma_N^2}{\lambda \sigma_N^2 },$ where $\mathbb{E}[\mathcal{P}(H_1)]=\bar{P}$. Also, define
  $  L_b^{(6)}=\sum_{i=1}^{2}\frac{-\mathcal{P}_{\text{WF}}(|H_b|)||\mathbf{Z}_{b,i}||^2}{\lambda \sigma_N^2 }.$  Denoting  $\tilde{L}_b =  {n_c} C(G_b^2)+L_b^{(2)}+L_b^{(5)}+L_b^{(6)}$ and  $S_B(\widetilde{L})=\sum_{b=1}^B\widetilde{L}_b$, define $\widetilde{\mathcal{I}}_\gamma=\{S_B(\tilde{L})\leq \log \gamma\}$. Using Claim \ref{Claim:L_bound} in Appendix \ref{App:Proof_PeakPow},  we obtain $S_B(L)\leq S_B(\widetilde{L})$. From this inequality, we have
$\mathbb{P}[\mathcal{I}_\gamma']\geq\mathbb{P}[\widetilde{\mathcal{I}}_\gamma].$
 \par Next, as in the \textbf{PP} case, we intersect the event $\widetilde{\mathcal{I}}_\gamma$ with certain \emph{high probability} events so as to further  lower bound $\mathbb{
 P}[\widetilde{\mathcal{I}}_\gamma]$ appropriately. Towards defining the events, consider mutually independent collection of  i.i.d. $\mathcal{N}(0,1)$ random variables $\{U_b^{(i)},~b\in[1:B]\}$, for  $i\in\{1,2,3\}$. Define $W_B^{(1)}$ and $\mathcal{E}_2$ as in Lemma \ref{Lem:P_I_gamma_LB_PP_UB}.
  Further, let $$W_B^{(3)}=\sum_{b=1}^B\left(\left(L_b^{(5)}-\mathbb{E}\left[L_1^{(5)}\right]\right)\Big/\sqrt{\mathbb{V}\left[L_1^{(5)}\right]}\right).$$ For $c_3>0$ (depending on the other parameters), denote $
\mathcal{E}_4=\left \{W_B^{(3)}\leq \sum_{b=1}^BU_b^{(2)}+2c_3\log B\right \}.$ 
Similarly, for $L_b^{(6)}$ as defined previously, define $$W_B^{(4)}=\sum_{b=1}^B\left(\left(L_b^{(6)}-\mathbb{E}\left[L_1^{(6)}\right]\right)\Big/\sqrt{\mathbb{V}\left[L_1^{(6)}\right]}\right).$$  
   For some $c_4>0$ (depending on the other parameters), let $\mathcal{E}_5=\left \{W_B^{(4)}\leq \sum_{b=1}^BU_b^{(3)}+2c_4\log B\right \}.$  
 In order to show that the event $\mathcal{E}_2 \cap \mathcal{E}_4 \cap \mathcal{E}_5$ is indeed a high probability event, we make use of the strong approximation principle of partial sum of i.i.d. random variables. Here, we make use of the additional assumption that $\mathbb{E}[\mathcal{P}^{(2+\delta)}(|H_1|))]]<\infty$, for any small $\delta>0$, in the statement of Theorem \ref{Th:AvgPow_Bnds}. Then, using the arguments invoked in Lemma \ref{Lem:P_I_gamma_LB_PP_UB}, using strong approximation of partial sum of i.i.d. random variables (\cite{dasgupta2008asymptotic}, Theorem 12.7), there exist positive constants $c_7,~c_8$ such that $\mathbb{P}\left[\mathcal{E}_2 \cup \mathcal{E}_4^c\cup \mathcal{E}_5^c\right]\leq {c_7}/{B^{c_8}}.$ Now, invoking the same lines of argument as in Lemma \eqref{Lem:P_I_gamma_LB_PP_UB}, we obtain the required result. 
 \end{proof}
\section{Proof of Theorem \ref{Thm:R_BF_LBUB_FBL_EH}}
  \label{App:EH_UB}
We will make use of Lemma \ref{Lem:EH_BF_MetaConverse} in our proof.  Furthermore, we assume that the energy harvesting constraints hold with equality. This will  \emph{not impact} the second order term (see \cite{polyanskiy2010channel}, \cite{fong2018achievable} for this trick known as the Yaglom map trick). We also assume that the energy harvesting constraint holds for $[B,n_c]$\textsuperscript{th} slot only, i.e., $||\mathbf{X'}||^2_2\leq ||\mathbf{E}||_1$ a.s..  This is a relaxation of the constraint in Section \ref{SubSubSec:EH_Trans} and hence can only improve the rate.
	 To proceed, we need the following definitions. Crucially, we rely on the notation introduced in Section \ref{Sec:Model_Notation}. Let $$L_b^{(1)}=\sum_{i=1}^{2}\frac{H_{b,i}^2||\mathbf{X'}_{b,i}||_2^2}{\sigma_N^2\left(1+G_b^2\right)},~ L_b^{(2)}=\sum_{i=1}^{2}\frac{H_{b,i}\langle\mathbf{X'}_{b,i},\mathbf{Z}_{b,i}\rangle}{\sigma_N^2\left(1+G_b^2\right)},$$ $$L_b^{(3)}=-\sum_{i=1}^{2}\frac{G_b^2||\mathbf{Z}_{b,i}||_2^2}{\sigma_N^4\left(1+G_b^2\right)},~L_b= {n_c}C(G_b^2)+\sum_{\ell=1}^3L_b^{(\ell)}.$$ Using the notation from Section \ref{Sec:Model_Notation}, $\mathbf{S}_B(L)=\sum_{b=1}^BL_b$ and $\mathbf{S}_B(L')=\sum_{b=1}^BL_b'$, where, $$L_b' =  {n_c} C(G_b^2)+\frac{\sigma_N^2||\mathbf{E}_b||_1-\mathcal{P}_{\text{WF}}(|H_b|)||\mathbf{Z}_{b}||_2^2}{\lambda \sigma_N^2}+L_b^{(2)}.$$ Let $\mathcal{E}=\{\overline{L}_B\leq \log \gamma_n\}$ and $\mathcal{E}_0=\{\overline{L}_B'\leq \log \gamma_n\}$. For $i\in\{1,2,3\}$, let $\{U_b^{(i)},~b\in[1:B]\}$ be mutually independent collection of  i.i.d. $\mathcal{N}(0,1)$ random variables.
	\par Define $W_B^{(1)}=\left[||C(\mathbf{G}^2)||_1-Bn_c\mathbf{C}(\overline{E})\right]/ \sqrt{\mathbb{V}[n_cC(G_1^2)]}$, where the notation $C(\mathbf{G}^2)$ is as mentioned in Section \ref{Sec:Model_Notation}.  For some positive constant $c_1$ (depending on other parameters), define $$\mathcal{E}_1= \{W_B^{(1)}\leq \sum_{b=1}^BU_b^{(1)}+2c_1\log B\}.$$ Let $W_B^{(2)}= \left[||\mathbf{E}||_1-Bn_c\overline{E}\right]/\sqrt{\lambda\sigma_E^2}$. For $c_2>0$ (depending on the other fixed parameters), define   $$\mathcal{E}_2=\{W_B^{(2)}\leq \sum_{b=1}^BU_b^{(2)}+2c_2B^{1/4}\}.$$ Let $L_b^{(4)}={\left[\mathcal{P}_{\text{WF}}(|H_b|)||\mathbf{Z}_{b}||_2^2-n_c\overline{E}\right]}/{\lambda \sigma_N^2 }$ and $W_B^{(3)}=\sum_{b=1}^BL_b^{(4)}/[\mathbb{V}[L_1^{(4)}]]^{1/2} $. For $c_3>0$ (depending on other parameters), let $$\mathcal{E}_3=\{W_B^{(3)}\geq \sum_{b=1}^BU_b^{(3)}-2c_3\log B\}.$$ 
	\par From the fact that $\mathbb{E}\left[|H_1|^2\right]<\infty$, observe there exists $t_1>0$ such that $\mathbb{E}[\exp(tC(G_1^2))]<\infty$, for  $|t|\leq t_1$. Similarly, since $\mathcal{P}_{\text{WF}}(\cdot)\leq \lambda$, there exists $t_2>0$ such that  $\mathbb{E}[\exp(tL_1^{(4)})]<\infty$, for  $|t|\leq t_2$. By assumption, $\mathbb{E}[E_{[1,1]}^4]<\infty$. Hence, using strong approximation of partial sum of i.i.d. random variables (see \cite{csorgo2014strong}, Theorem 2.6.2 and \cite{komlos1976approximation}, Theorem 4), there exist positive constants $c_4,~c_5$ such that $
	\mathbb{P}\left[\bigcup_{i=1}^3\mathcal{E}_i^c\right]\leq {c_4}/{B^{c_5}}.$
	\par Finally, let $\nu_1\triangleq nn_c\mathbb{V}[C(G_1^2)]+ n\sigma_E^2/\lambda^2+B\mathbb{V}[L_1^{(4)}]$ and $\nu_2\triangleq \nu_1+\sum_{b=1}^{B}{2L_b^{(1)}}/{(1+G_b^2)}$.  Now, we outline the key steps of the proof. Consider a scheme, i.e. distribution $P_{\mathbf{X'|E,H}}$, that satisfies the \emph{relaxed energy harvesting constraint} (mentioned at the beginning of the proof)  and attains an average probability of error no greater than $\epsilon$. We fix the auxiliary channel as in Section \ref{SubSubec:Ch_Aux_PP_UB}. With this choice, we consider lower bounding the probability term in \eqref{Eqn:Meta_Converse}. That is, $\mathbb{P}\left[\log\frac{d\mathbb{P}_{\mathbf{Y|X',H}}}{d\mathbb{Q}_\mathbf{Y|H}} \leq \log \gamma_n\right] \stackrel{(a)}{=} \mathbb{P}\left[\mathcal{E} \right]\stackrel{(b)}{\geq}\mathbb{P}\left[\mathcal{E}_0 \right].$ Here, $(a)$ follows from the fact that the corresponding random variables (in the definition of the events) are equal in distribution. Next, $(b)$ follows from Claim \ref{Claim:L_bound} in Appendix \ref{App:Proof_PeakPow} and the assumption of relaxed energy harvesting constraint. Now, we further lower bound $\mathbb{P}\left[\mathcal{E}_0 \right]$ using the strong approximation result mentioned above. Specifically, note that $\mathbb{P}\left[\mathcal{E}_0 \right] \geq \mathbb{P}\left[  \bigcap_{i=0}^{3} \mathcal{E}_i \right]$. However, $$\mathbb{P}\left[  \bigcap_{i=0}^{3} \mathcal{E}_i \right] \ge \Phi\left( \leq \frac{\log \gamma_n-n\mathbf{C}(\overline{E})}{\sqrt{\nu_2}}\right)-\delta_n.$$ Here, with  $\delta_n={(c_4n_c^{c_5})}/{n^{c_5}}$, the bound  follows from strong approximation bound $\mathbb{P}\left[\bigcup_{i=1}^3\mathcal{E}_i^c\right]\leq {c_4}/{B^{c_5}}.$ Next, we observe $$ \Phi\left( \leq \frac{\log \gamma_n-n\mathbf{C}(\overline{E})}{\sqrt{\nu_2}}\right) \ge \Phi\left( \frac{\log \gamma_n-n\mathbf{C}(\overline{E})}{\sqrt{nV_{\text{EF}}''(\overline{E})}}\right)$$ if we ensure $\log \gamma_n<n\mathbf{C}(\bar{E})$ and noting that $\nu_2\geq nV_{\text{EF}}''(\overline{E})$ ($V_{\text{EF}}''(\overline{E})$ is as defined in the statement of Theorem \ref{Thm:R_BF_LBUB_FBL_EH}). Finally, $\Phi\left( \frac{\log \gamma_n-n\mathbf{C}(\overline{E})}{\sqrt{nV_{\text{EF}}''(\overline{E})}}\right)-\delta_n\geq \epsilon+\delta_n$ is obtained by choosing $\log \gamma_n=n\mathbf{C}(\overline{E})+\sqrt{nV_{\text{EF}}''(\overline{E})}\Phi^{-1}(\epsilon+2\delta_n)$. Note that, as required, $\log \gamma_n<n\mathbf{C}(\overline{E})$ can be ensured for appropriately chosen $n$, as $\epsilon<1/2$. Thus, we have shown $\mathbb{P}[\mathcal{E}_0]\geq \epsilon+\delta_n$. The result follows  by plugging in the above estimate and $\gamma_n$ in \eqref{Eqn:Meta_Converse}, and applying Taylor's theorem.\qed  
\section{Proof of Theorem \ref{Thm_mod_dev}}
\label{App:Mod_Deviation}
First, with $a_n$ as in Definition \ref{Defn:Mod_Dev_Prop}, we consider lower bounding  $\underline{\lim}_{n\rightarrow \infty}\frac{1}{na_n^2}\log  p_{e,\text{avg}}^*(n,M_n)$ under \textbf{PP} constraint.  Consider any $(n,M_n,\epsilon_n,\bar{P})$ code with $M_n=\exp({n(\mathbf{C}(\bar{P})-a_n)})$. Fix $\theta>1$ and let $\gamma=\exp({n(\mathbf{C}(\bar{P})-\theta a_n)})$. Consider the auxiliary channel  in Section \ref{SubSubec:Ch_Aux_PP_UB}. Invoking the meta converse bound and applying the relaxations as in Appendix \eqref{App:Proof_PeakPow} (see Section  \ref{SubSubec:UB_via_MetaConv_PP_UB} therein), we obtain $p_{e,\text{avg}}^*(n,M_n) \geq \mathbb{P}[\mathcal{I}_\gamma]-\exp(na_n(\theta-1)).$  Using a similar  analysis as in Appendix \ref{SubSec:UB_PP} that involves lower bounding the probability term $\mathbb{P}[\mathcal{I}_\gamma]$ in terms of the probability of sum of i.i.d. standard Gaussian random variables, and invoking the result from \cite{dembo2010large} (Theorem 3.7.1),  we obtain the required lower bound as in \cite{polyanskiy2010channelmode}. A similar analysis yields $V_{\text{EF}}''(\overline{E})$ as lower bound on the moderate deviation constant under \textbf{EH} constraint.  With an additional moment constraint on the power allocation function (as in the statement of Theorem \ref{Th:AvgPow_Bnds}), we obtain $V_{\text{BF}}(\bar{P})$ as the lower bound for the \textbf{AP} constraint as well.
\par For the lower bound for \textbf{PP}, \textbf{AP} and \textbf{EH} constraint, the $\beta\beta$ bound and the power control policies explained in Appendix \ref{SubSec:LB_PP}, \ref{SubSec:LB_AP} and  \ref{SubSec:Save_Transmit} respectively. We fix the same auxiliary channel as in the analysis therein. For appropriate $\delta_n>0$, with i.i.d. $\mathcal{N}(0,1-\delta_n)$ input in the achievability proof for \textbf{PP}, we observe that there exists $(n,M_n,\epsilon_n,\bar{P})$ codes such that $\log M_n^*(n,\epsilon_n,\bar{P})\geq -\log \gamma_0 +O(1)$ (see the proof in Appendix \ref{SubSec:LB_PP}). Choose $\log \gamma_0=n\mathbf{C}(\bar{P})-\theta na_n$, $0<\theta<1$ so that we obtain the existence of codes with $\log M_n^*(n,\epsilon_n,\bar{P})\geq n\mathbf{C}(\bar{P})-\theta na_n $, for $n$ large, by [\cite{dembo2010large}], Theorem 3.7.1],  $\limsup\limits_{n\rightarrow\infty}\log \frac{p_{e,\text{avg}}^*(n,M_n)}{n}\leq -\frac{\theta^2}{2V_{\text{BF}}(\bar{P})}$. Taking $\theta\nearrow 1$, we obtain the required result. A similar analysis yields the result for the \textbf{AP} and \textbf{EH} constraint case as well.
\qed
  \bibliographystyle{IEEEtran}
\bibliography{Delay_PowControl} 
 \end{document}